\documentclass[11pt, reqno]{article}
\pdfoutput=1

\usepackage{graphicx}
\usepackage{jheppub}
\usepackage{amssymb}
\usepackage{amsmath}
\usepackage[usenames,dvipsnames]{xcolor}
\usepackage{epsfig}
\usepackage{dcolumn}
\usepackage{tikz}
\usetikzlibrary{shapes.geometric, arrows}
\usepackage{upgreek}
\usepackage{setspace}
\usepackage{lmodern}
\usepackage{enumitem}
\usepackage{array,multirow,bigdelim,arydshln}
\usepackage{appendix}
\usepackage{xparse}
\usepackage[utf8]{inputenc}
\usepackage{hyperref}
\usepackage[dvipsnames]{xcolor}
\usepackage{blkarray}
\hypersetup{
	colorlinks,
	urlcolor=Maroon,
	linkcolor=Maroon,
	citecolor=Maroon
}

\usepackage{amsthm}
\newtheorem{thm}{Theorem}[section]

\newtheorem{lem}[thm]{Lemma}

\newtheorem{defn}[thm]{Definition}
\theoremstyle{definition}

\usepackage{mathtools}

\usepackage{float}
\restylefloat{table}

\NewDocumentCommand{\binomial}{omm}
{%
	\genfrac(){0pt}{}{#2}{#3}%
	\IfValueT{#1}{_{\!#1}}%
}
\NewDocumentCommand{\eulerian}{omm}
{%
	\genfrac<>{0pt}{}{#2}{#3}%
	\IfValueT{#1}{_{\!#1}}%
}

\def \s {\sigma}

\usepackage{underscore}

\usepackage{latexsym}
\usepackage{tikz}

\title{Smoothly Splitting Amplitudes and Semi-Locality}

\author[a]{Freddy Cachazo,}
\author[b]{Nick Early}
\author[a,c]{and Bruno Gim\'enez Umbert}
\affiliation[a]{Perimeter Institute for Theoretical Physics, Waterloo, ON N2L 2Y5, Canada.}
\affiliation[b]{Max Planck Institute for Theoretical Physics, F\"{o}hringer Ring 6, Munich, Germany.}
\affiliation[c]{Department of Physics and Astronomy, Western University, London, ON N6A 5B7, Canada.}
\emailAdd{fcachazo@pitp.ca}
\emailAdd{earlnick@mpp.mpg.de}
\emailAdd{bgimenez@uwo.ca}

\abstract{ In this paper, we study a novel behavior developed by certain tree-level scalar scattering amplitudes, including the biadjoint, NLSM, and special Galileon, when a subset of kinematic invariants vanishes without producing a singularity. This behavior exhibits properties which we call \textit{smooth splitting} and \textit{semi-locality}. The former means that an amplitude becomes the product of exactly three amputated Berends-Giele currents, while the latter means that any two currents share one external particle.  
	We call these smooth splittings 3-splits.  In fact, there are exactly $\binom{n}{3}-n$ such 3-splits of an $n$-particle amplitude, one for each tripod in a polygon; as they cannot be obtained from standard factorization, they are a new phenomenon in Quantum Field Theory. In fact, the resulting splitting is analogous to the one first seen in Cachazo-Early-Guevara-Mizera (CEGM) amplitudes which generalize standard cubic scalar amplitudes from their ${\rm Tr}\, G(2,n)$ formulation to ${\rm Tr}\, G(k,n)$, where ${\rm Tr}\, G(k,n)$ is the tropical Grassmannian. Along the way, we show how smooth splittings naturally lead to the discovery of mixed amplitudes in the NLSM and special Galileon theories and to novel BCFW-like recursion relations for NLSM amplitudes.
}

\begin{document}
	\maketitle
	\addtocontents{toc}{\protect\setcounter{tocdepth}{1}}
	\def \tr {\nonumber\\}
	\def \nn {\nonumber}
	\def \la {|}
	\def \ra {|}
	\def \dd {\Theta}
	\def\hset{\texttt{h}}
	\def\gset{\texttt{g}}
	\def\sset{\texttt{s}}
	\def \be {\begin{equation}}
		\def \ee {\end{equation}}
	\def \ba {\begin{eqnarray}}
		\def \ea {\end{eqnarray}}
	\def \k {\kappa}
	\def \h {\hbar}
	\def \r {\rho}
	\def \l {\lambda}
	\def \be {\begin{equation}}
		\def \en {\end{equation}}
	\def \bes {\begin{eqnarray}}
		\def \ens {\end{eqnarray}}
	\def \red {\color{Maroon}}
	\def \pt {{\rm PT}}
	\def \s {\textsf{s}}
	\def \t {\textsf{t}}
	\def \C {\textsf{C}}
	\def \tp {||}
	\def \p {x}
	\def \x {z}
	\def \V {\textsf{V}}
	\def \ls {{\rm LS}}
	\def \ma {\Upsilon}
	\def \SL {{\rm SL}}
	\def \GL {{\rm GL}}
	\def \w {\omega}
	\def \e {\epsilon}
	
	\numberwithin{equation}{section}
	

	
	
	\section{Introduction}\label{sec:introduction}
	
	
	Unitarity and locality are the basic pillars of quantum field theory. Using them as constraints on the S-matrix allows for the construction of scattering amplitudes using recursion relations such as the Berends-Giele \cite[p.23]{Dixon:1996wi} or BCFW techniques \cite{Britto:2005fq,Britto:2004ap,Britto:2004nc}. At tree-level\footnote{Instead of restricting to tree-level, the correct way to describe this is by saying that one-particle states in the completeness relation imply the presence of poles. In this paper we only work at tree-level so the restriction is enough.}, unitarity implies that scattering amplitudes have simple poles of the form $1/(P^2-m^2+i\epsilon)$, with $P$ the sum of momenta of a subset of particles participating in the process and $m$ the mass of a particle in the spectrum of the theory. Moreover, the residues are also completely determined to be the product of two smaller scattering amplitudes; this property is called {\it factorization}. The original set of particles is partitioned into two sets, often called ``left" and ``right". The two amplitudes in the residue only share an ``internal" particle, with momentum $P$, which is taken to be on-shell, i.e. $P^2=m^2$. Now, locality is the statement that tree-level amplitudes do not have any other kind of singularities which makes clear the power of the constraints in their computation. There is an important caveat; unitary and locality constrain singularities at finite momenta and unless emergent symmetries at large momenta are present \cite{Arkani-Hamed:2008bsc}, there could be singularities at infinite momenta which prevent the complete reconstruction of the amplitude \cite{Benincasa:2007xk, elvanghuang2015}.
	
	Scattering amplitudes in theories with color/flavour are dramatically simplified by the color decomposition into partial amplitudes \cite[p.4]{Dixon:1996wi}. The main reason for the simplification is that each partial amplitude can only contain a certain subset of all possible poles the full amplitude can have. These are called {\it planar poles}. For the canonical ordering $\mathbb{I} = (1,2,\ldots ,n)$, the only possible poles are of the form $1/(p_i+p_{i+1}+\ldots +p_{i+m})^2$. From now on we restrict our attention to massless theories and only to the canonical ordering. Therefore we will simply refer to it as the planar ordering. 
	
	Conventional wisdom would say that in regions where kinematic invariants of the form $(p_i+p_j)^2$ vanish and are not planar then a partial amplitude would not have any interesting behavior.   
	
	In this work we find that in fact there are subspaces in the space of kinematic invariants where some non-planar kinematic invariants vanish and the partial amplitude becomes the product of lower point objects without becoming singular. We call the resulting behavior of amplitudes a {\it smooth splitting} and the corresponding subspace of kinematics invariants {\it split kinematics}.
	
	Unlike standard factorizations, smooth splittings are semi-local, i.e., a particle can participate in two of the factors. Each factor is not an amplitude but an amputated Berends-Giele current (see e.g. \cite{Mafra:2016ltu}) as they possess one emerging leg which is off-shell. The current is said to be amputated because the propagator corresponding to the off-shell leg is not present.
	
	We find that amplitudes factor in exactly three pieces and we call the corresponding behaviour a 3-splitting. When one of the amputated currents only has two on-shell external legs, it becomes trivial, i.e. it is a constant. In such degenerate 3-splits, all original particles but one are either in the ``left" or ``right" currents while exactly one on-shell external particle is in both. In general 3-splits, each pair of currents shares an external on-shell particle. We call this phenomenon {\it semi-locality}. 
	
	A very important property of 3-splits is that they cannot be obtained from standard unitarity or factorization arguments and thus they do not have any close analog within the standard QFT literature. However, in the recent generalization of QFT amplitudes known as CEGM amplitudes, analogous 3-split behavior is common but it appears as the residue of a pole.  
	
	A simple way to define split kinematics is by using the structure of the matrix of Mandelstam invariants with entries $s_{a,b}$. Start by introducing three rows and columns labeled by $(i,j,k)$ with $i<j<k$ with non-zero entries and not all three labels cyclically adjacent. Without loss of generality we often set $i=1$. This gives the matrix a ``tic-tac-toe" structure. In other words, the matrix now has nine chambers. Split kinematics simply sets to zero the elements $s_{a,b}$ in the six non-diagonal chambers. A schematic representation is given in Figure \ref{fig:intromatdensityplot}. 
	\begin{figure}[h!]
		\centering
		\includegraphics[width=0.45\linewidth]{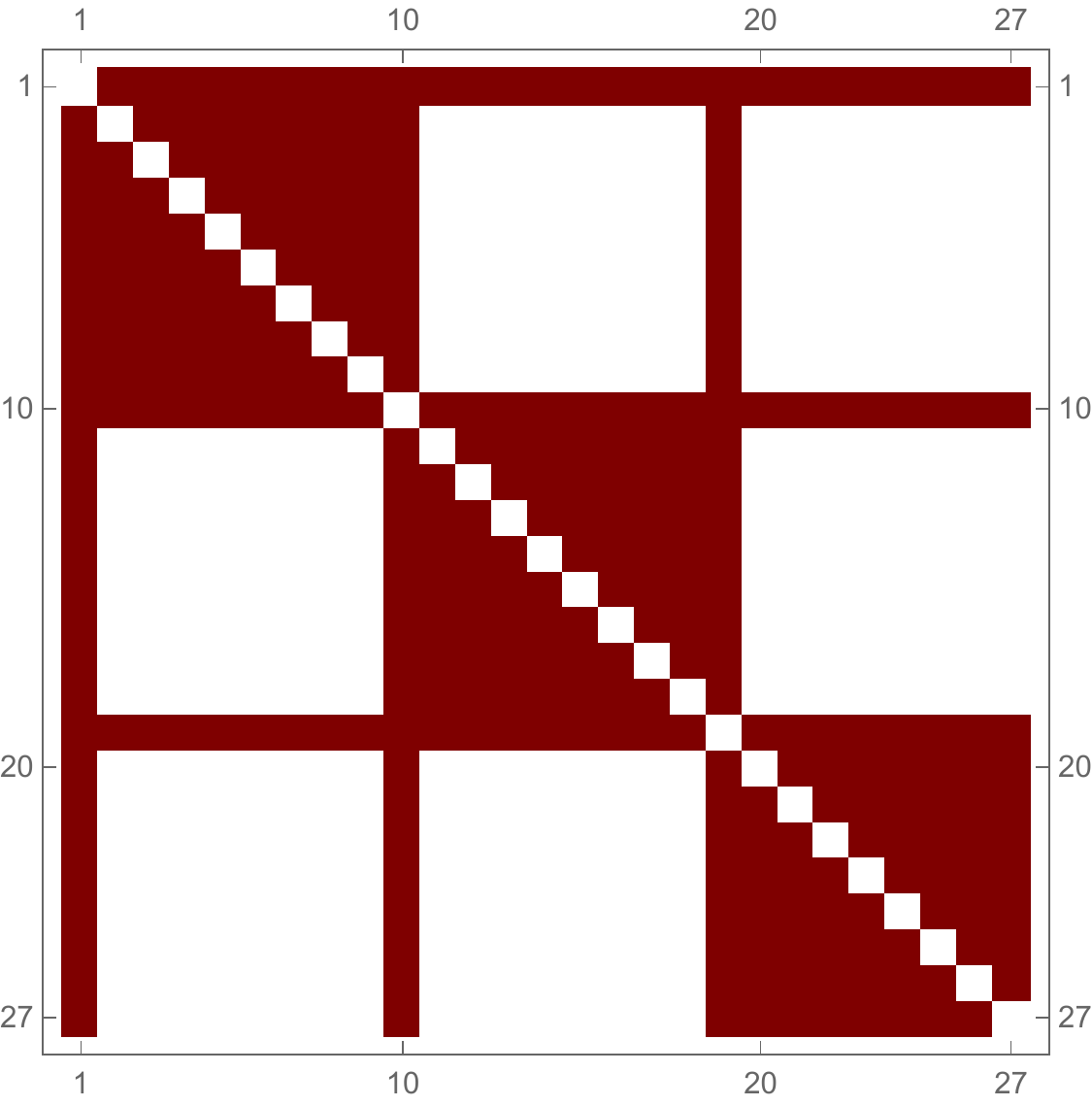}
		\caption{Matrix of Mandelstam invariants $s_{ab}$ for the $n=27$ split kinematics $(1,10,19)$, having set to zero all $s_{ab}$'s in the unshaded regions.}
		\label{fig:intromatdensityplot}
	\end{figure}
	Here the dark entries are non-zero while white entries are zero. Of course, since we are dealing with massless particles any invariant of the form $s_{a,a}$ is zero and hence the white diagonal.  
	
	There are exactly $\binom{n}{3}-n$ 3-splits. Incidentally, this is also the dimension of the space of kinematic invariants of generalized $k=3$ CEGM amplitudes and in particular it is the number of planar basis elements, each of which characterizes a pole \cite{Early:2019eun}. Degenerate 3-splits are achieved when two of the labels $(i,j,k)$ are consecutive in the canonical ordering.
	
	The simplest non-trivial 3-split is obtained from the $n=6$ biadjoint partial amplitude $m_6(\mathbb{I} ,\mathbb{I})$ under the $(1,3,5)$ split kinematics, i.e, by setting $s_{24}=s_{46}=s_{62}=0$ and the result is given by
	\be\label{example}
	\left. m_6(\mathbb{I} ,\mathbb{I}) \right|_{\rm split\, kin.} = \left(\frac{1}{s_{12}}+\frac{1}{s_{23}}\right)\left(\frac{1}{s_{34}}+\frac{1}{s_{45}}\right)\left(\frac{1}{s_{56}}+\frac{1}{s_{61}}\right).
	\ee
	This example is discussed in detail in Section \ref{sec: smooth splits intro}. Here we only point out the semi-local character of the expression since particle $3$ participates in the first and second factor, particle $5$ in the second and third, and particle $1$ in the third and first. It is also clear that any of the three factors are obtained from a four-point (three on-shell) Berends-Giele current by amputating the propagator of the off-shell leg. For example, in the first factor the off-shell leg has momentum $P_I = -p_1-p_2-p_3$ and $P_I^2 = s_{123}\neq 0$.

	\subsection{Main Results}
	
	Now we present the main results of this work. We show that on the $(i,j,k)$ split kinematics subspace, the biadjoint partial amplitude becomes a product
	\be\label{eq: 3split biadjoint }
	\left. m_n(\mathbb{I} ,\mathbb{I}) \right|_{\rm split\, kin.} = \mathcal{J}(i,i+1,\ldots,j-1,j){\cal J}(j,j+1,\ldots,k-1,k)\mathcal{J}(k,k+1,\ldots,i-1,i).
	\ee
	We further show that $\mathcal{J}(1,2,\ldots, m)$ denotes an amputated current with the off-shell leg carrying momentum $-(p_1+p_2+\cdots +p_m)$. 
	
	Surprisingly, we find that not only biadjoint amplitudes exhibit smooth splitting but so do non-linear sigma model (NLSM) \cite{Kampf:2013vha} and special Galileon amplitudes \cite{Cachazo:2014xea}. The special Galileon theory does not have color/flavor ordering and thus it seems to be outside the scope of the construction. However, the derivative interactions manage to keep the amplitude finite in the limit as split kinematics is approached. In both, NLSM and special Galileon amplitudes, smooth splitting produces currents in their corresponding {\it extended theories}, as defined by Cachazo, Cha, and Mizera (CCM) \cite{Cachazo:2016njl}. In fact, the smooth splitting behavior provides a new approach to discover the extended theories without resorting to soft limits.  
	
	In very schematic form, NLSM and special Galileon amplitudes split as follows,
	\be\label{main2}
	\left. A_n^{\rm NLSM}(\mathbb{I})\right|_{\rm split\, kin.} = \mathcal{J}^{\rm NLSM}(\mathbb{I})\; \mathcal{J}^{{\rm NLSM}\oplus \phi^3}(\mathbb{I}_1|\beta_1)\;\mathcal{J}^{{\rm NLSM}\oplus \phi^3}(\mathbb{I}_2|\beta_2)
	\ee
	and
	\be\label{main3}
	\left. A_n^{\rm sGal}\right|_{\rm split\, kin.} = \mathcal{J}^{\rm sGal}\; \mathcal{J}^{{\rm sGal}\oplus \phi^3}(\beta_1)\; \mathcal{J}^{{\rm sGal}\oplus \phi^3}(\beta_2)
	\ee
	where $\mathbb{I}_i$ and $\beta_i$ are planar orderings of certain subsets of particle labels. These formulas are derived in Section \ref{sec: smooth splits others}.
	
	As explained above, having knowledge of the behavior of amplitudes in regions of the space of kinematic invariants can be used to partially or totally reconstruct them. Split kinematics provides novel regions that can be used in addition to unitarity to constrain amplitudes. In fact, NLSM amplitudes are examples where standard recursive techniques do not work \cite{Kampf:2013vha}. This motivated the use of soft-limits in their recursive construction \cite{Cheung:2015ota}. Here we show how smooth splitting leads to novel BCFW relations for NLSM amplitudes that do not require knowledge of soft limits.   
	
	More precisely, the new BCFW construction induces split kinematics at two points on the one-dimensional deformation space, chosen to be $z=1$ and $z=-1$. We prove that the following formula,
	$$ A^{\rm NLSM}_n(\mathbb{I}) = \frac{1}{2\pi i}\oint_{|z|=\epsilon} dz \frac{A^{\rm NLSM}(z)}{z(1-z^2)},$$
	provides a recursion relation without contributions at infinity. As an example we obtain a new formula for the six-point amplitude
	\be
	\begin{split}
		A^{\rm NLSM}_6(\mathbb{I}) = & \frac{(s_{12}+s_{23}-s_{123})(s_{45}+s_{56})}{s_{123}}+  \frac{(s_{23}+s_{34})(s_{56}+s_{61}-s_{234})}{s_{234}}+\\ & \frac{(s_{34}+s_{45}-s_{345})(s_{61}+s_{12}-s_{345})}{s_{345}}+s_{34}+s_{45}-s_{345}.
	\end{split}
	\ee

	This work is organized as follows.  We start in Section \ref{sec: smooth splits intro} with examples that motivate and illustrate smooth splitting and split kinematics. 
	This kinematics generically leads to smooth 3-splits but we also point out the border cases in which it produces smooth 3-splits with only two nontrivial factors. In Section \ref{sec: smooth splits biadjoint scalar}, we study 3-splits in the biadjoint theory. We prove the general formula in terms of three amputated currents using the CHY formalism. In Section \ref{sec: smooth splits others}, we consider NLSM and special Galileon amplitudes. In Section \ref{sec: recursions}, we use smooth splittings to derive novel BCFW-like recursion relations for NLSM amplitudes in which soft limits are not required. In Section \ref{sec:discussion} we discuss relations to soft limits, soft triangulations, CEGM amplitudes, how to smoothly split currents, and generalizations to other theories.

	\section{Split Kinematics}\label{sec: smooth splits intro}
	
	The purpose of this section is to give a presentation of our main results with the biadjoint theory as example; further discussion and proofs are given in subsequent sections, including extensions to other theories.
	
	We first introduce split kinematics, in Definition \ref{defn: split kinematics} for Mandelstam invariants $s_{i,j}$ and for planar basis invariants $s_{i,i+1,\ldots, j}$, using a homotopy rule on directed arcs in a disk with (counterclockwise)-oriented boundary, in Figure \ref{fig:equivalencerelationarcs} below.
	We shall always assume that a triple $(i,j,k)$ has the cyclic order $i<j<k$.
	\begin{defn}\label{defn: split kinematics}
		Given any triple of distinct indices $i,j,k$ that is not a (cyclic) interval in $\{1,\ldots,n\}$ of the form $(a,a+1,a+2)$, say with with $1\le i<j<k \le n$, the \textit{split\footnote{The term \textit{split} kinematics is reminiscent of a particular kind of matroid subdivision, the d-split, which decomposes a polytope into $d$ maximal cells with a condition on the way their internal faces intersect.  In this paper we show that on split kinematics, certain amplitudes decompose as two-fold or three-fold products.  The analogy is motivated in part by  CEGM amplitudes in light of Equation \eqref{eq:3n3split}, where the planar kinematic invariant $\eta_{246}$ is literally dual to a 3-split matroid subdivision with three maximal cells.} kinematics} subspace is characterized by the following condition: $s_{a,b} = 0$ whenever the pair $(a,b)$ interlaces the triple $(i,j,k)$, having modulo cyclic rotation
		\begin{eqnarray}
			& & a<i<b<j<k,\ \text{or}\nonumber\\
			& & a<j<b<k<i,\ \text{or}\\
			& & a<k<b<i<j.\nonumber
		\end{eqnarray}
	\end{defn}
	For example, for $n=7$ particles and the triple (1,3,6), the split kinematics subspace is cut out by the equations
	$$s_{24} = 0,\ s_{25} = 0,\ s_{27} = 0,\ s_{47} = 0,\ s_{57}=0.$$  
	
	
	We first formulate the notion of a smooth split in full generality, in the prototypical case, the biadjoint scalar partial amplitude $m_n(\mathbb{I},\mathbb{I})$, and then show that only the case of a smooth 3-split can be achieved by a suitable restriction of the amplitude to a subspace of the kinematic space.
	
	The definition of the biadjoint theory can be found e.g. in \cite{Cachazo:2014xea}. Here we only provide the definition of $m_{n}(\mathbb{I},\mathbb{I})$ as it is our main object of study.
	
	\begin{defn}\label{defn: biad}
		Let ${\cal T}$ be the set of all planar unrooted binary trees with $n$ leaves. A momentum vector, $p\in \mathbb{R}^{1,D-1}$ with $p^2 := p\cdot p=0$, is assigned to each leaf such that $s_{ab}=(p_a+p_b)^2$ and $p_1+p_2+\ldots +p_n=0$. Given a tree $T\in {\cal T}$, each edge $e$ of $T$ partitions the leaves into two sets $L_e\cup R_e = [n]$. Clearly, 
		$$\left( \sum_{a\in L_e} p_a \right)^2 = \left( \sum_{a\in R_e} p_a \right)^2 := P_e^2. $$
		The partial biadjoint amplitude then given by  
		\be
		m_{n}(\mathbb{I},\mathbb{I}) := \sum_{T\in {\cal T}}\frac{1}{\prod_{e\in E(T)}P_e^2},
		\ee
		where $E(T)$ is the edge set of $T$. 
	\end{defn}

	\begin{defn}\label{defn: d-split}
		For any $d\ge 2$, a smooth d-split is a decomposition 
		\begin{eqnarray}\label{eq:d-split}
			m_{n}(\mathbb{I},\mathbb{I})\vert_{\rm split\, kin.} = \mathcal{J}(j_1,\ldots, j_2) \mathcal{J}(j_2,\ldots, j_3)\cdots \mathcal{J}(j_{d},\ldots, j_1),
		\end{eqnarray}
		where each $\mathcal{J}(a,\ldots, b)$ is an amputated current, with exactly one off-shell leg.
	\end{defn}
	In fact, we claim that in Definition \ref{defn: d-split}, only the case $d=3$ is possible.  Assuming this for now, then we have the following cases.
	\begin{enumerate}
		\item No pair of labels is cyclically consecutive, that is
		$$\{(j_1,j_2),(j_2,j_3),(j_3,j_1)\} \cap \{(1,2),(2,3),\ldots, (n,1)\} = \emptyset.$$
		In this case, all three factors are nontrivial amputated currents. 
		\item If exactly one pair is cyclically consecutive, say $(j_3,j_1)=(j_3,j_3+1)$ (modulo $n$), then $\mathcal{J}(j_3,\ldots, j_1) = 1$ and Equation \eqref{eq:d-split} reduces to 
		\begin{eqnarray}\label{eq:d-split2}
			m_{n}(\mathbb{I},\mathbb{I})\vert_{\rm split\, kin.} = \mathcal{J}(j_1,\ldots, j_2) \mathcal{J}(j_2,\ldots, j_1-1).
		\end{eqnarray}
		\item If two pairs are cyclically consecutive, so the triple is a single cyclic interval $(j,j+1,j+2)$, then no condition is imposed on the kinematics.
	\end{enumerate}
	Evidently, the $\binomial{n}{3}-n$ nontrivial smooth 3-splits are in in bijection with the interior tripods in a polygon with cyclic vertices, or equivalently, collections of triples $(i,j,k)$ such that no two pairs of indices are cyclically adjacent.
	
	
	Now we show that only $d=3$ can be achieved by restricting $m_n(\mathbb{I},\mathbb{I})$ to some subspace of the kinematic space\footnote{Here we are assuming that the $s_{ab}$ are formal variables; later we shall specialize to the case when they are inner products of momentum vectors.}.

	The argument is very simple. First note that by Equation \eqref{defn: biad}, a biadjoint amplitude $m_{n}(\mathbb{I},\mathbb{I})$ must have degree $-(n-3)$ in the Mandelstam invariants (or in physics terminology, mass dimension $-2(n-3)$). Now, an amputated current is nothing but an amplitude with one leg off-shell, i.e. such that one of the corresponding momentum vectors does not have zero Minkowski norm. Therefore the degree of a current agrees with that of an amplitude with the same number of legs.   
	
	Supposing that we had a smooth d-split as in Equation \eqref{eq:d-split}, then the degree of the product would be
	\be\label{degreeC}
	-\sum_{t=1}^d\left(j_{t+1} - j_t - 1\right) = -(n - d),
	\ee
	where the indices are cyclic modulo $n$, which matches the degree of $m_n(\mathbb{I},\mathbb{I})$ only when $d=3$.  Thus, smooth d-splits cannot occur for $d\not= 3$ as such.  Of course, one can imagine possible generalizations of smooth 3-splits in this direction, but such questions are beyond the scope of this paper.  See Section \ref{sec: smoothly splitting currents} for some additional remarks.
	
	

	\subsection{Split Kinematics: Planar Poles }
	
	The rule for how planar poles $s_{i, i+1,\cdots j}$ decompose when restricted to split kinematics can be compactly formulated in a picture, as in Figures  \ref{fig:equivalencerelationarcs} and \ref{fig:part6135Intro}.
	
	All relations between planar basis elements that occur when imposing split kinematics have the form
	$$s_{a,a+1,\ldots, c} = s_{a,a+1,\ldots, b} + s_{b,b+1,\ldots, c},$$
	for $a<b<c$ modulo $n$.  To state the criterion determining which planar poles decompose in this way it is convenient to draw the $n$ indices $1,\ldots, n$ on the boundary of a disk, say counterclockwise.  Given a split $(i,j,k)$, then we connect the three legs $i,j,k$ in the center of the disk to form a \textit{tripod}, as in Figure \ref{fig:tripodexample}, which partitions the disk into three connected components which we identify with the amputated currents themselves.  
	\begin{figure}[h!]
		\centering
		\includegraphics[width=0.3\linewidth]{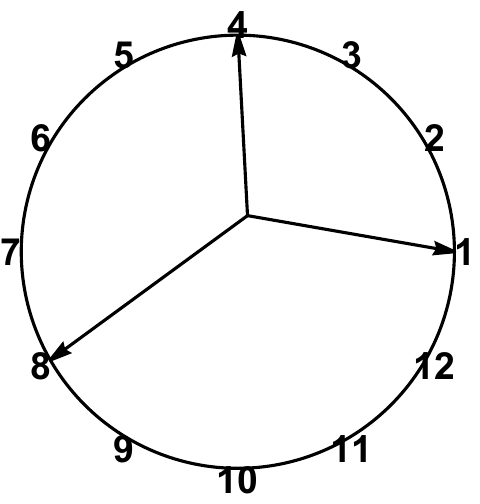}
		\caption{The tripod (1,4,8) for $n=12$ particles.}
		\label{fig:tripodexample}
	\end{figure}
	
	Thus, Figure \ref{fig:equivalencerelationarcs} can be straightforwardly generalized to determine the conditions under which such a decomposition takes place; see the Caption of Figure \ref{fig:equivalencerelationarcs} for details.
	\begin{figure}[h!]
		\centering
		\includegraphics[width=1\linewidth]{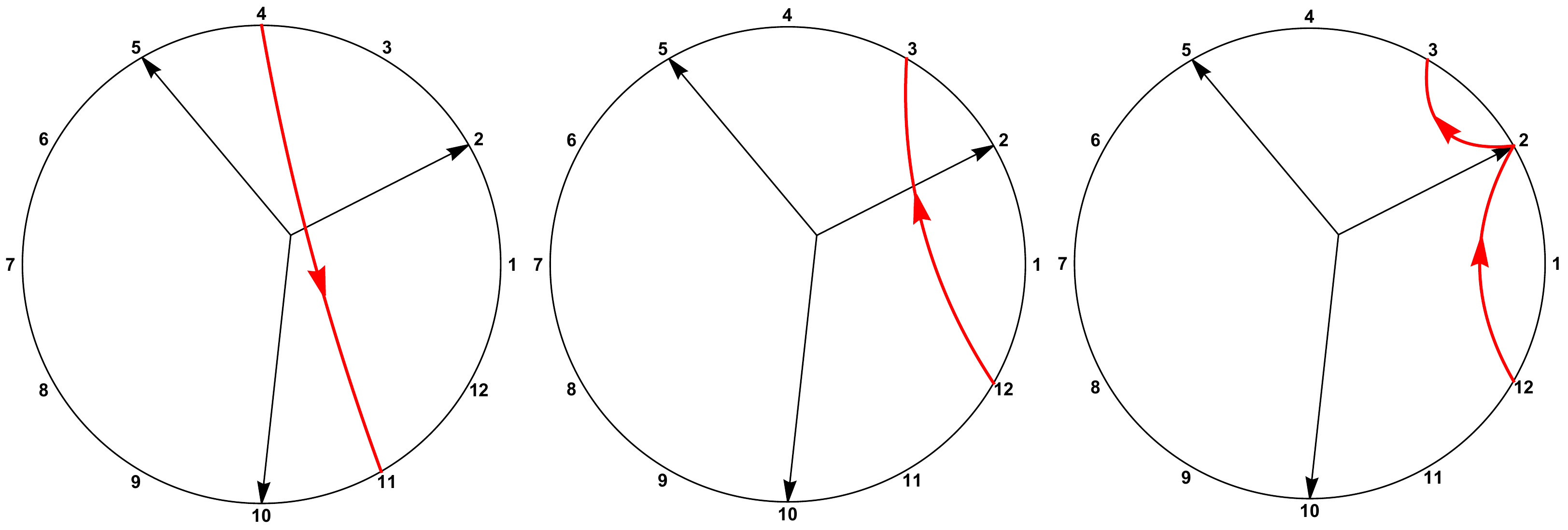}
		\caption{The boundary of the disk is always oriented counter-clockwise.  To determine the relations among planar poles imposed by split kinematics, draw any directed arc $(a \rightarrow c)$ and homotope the arc until it is near the parallel counter-clockwise oriented boundary of the disk.  If the arc has to cross exactly one leg of the tripod, say $b$ with $a<b<c$ cyclically, then the planar pole decomposes as $s_{a,a+1,\ldots,  c} = s_{a,a+1,\ldots, b} + s_{ b,b+1,\ldots,  c}$.  Above, due to momentum conservation, we have the relation $s_{a,a+1,\ldots, b} = s_{b+1,b+2,\ldots, a-1}$, noting that the directed arc correspondingly switches orientation in going from the left to the center diagrams.  Applying the homotopy to the parallel boundary of the disk gives that $s_{4,5,\ldots,11} = s_{12,1,2,3} = s_{12,1,2} + s_{2,3}$.}
		\label{fig:equivalencerelationarcs}
	\end{figure}
	In Figure \ref{fig:part6135Intro} we calculate the relations which define the split kinematics subspace $(1,3,5)$ in the planar basis.
	
	\begin{figure}[h!]
		\includegraphics[width=1\linewidth]{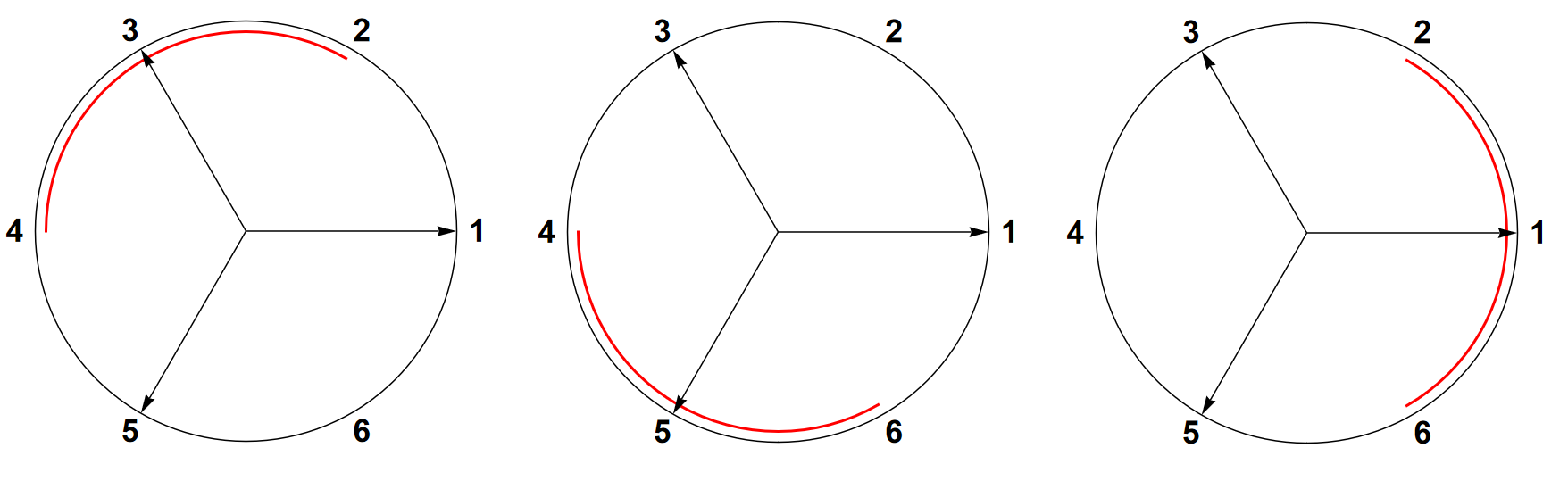}
		\centering
		\caption{Planar basis relations for the $n=6$ split kinematics given by the partition $(1,3,5)$.  Here $s_{234} = s_{23} + s_{34}$, $s_{456} = s_{45} + s_{56}$, and $s_{612} = s_{61} + s_{12}$.}
		\label{fig:part6135Intro}
	\end{figure}

	\begin{figure}[h!]
		\centering
		\includegraphics[width=0.5\linewidth]{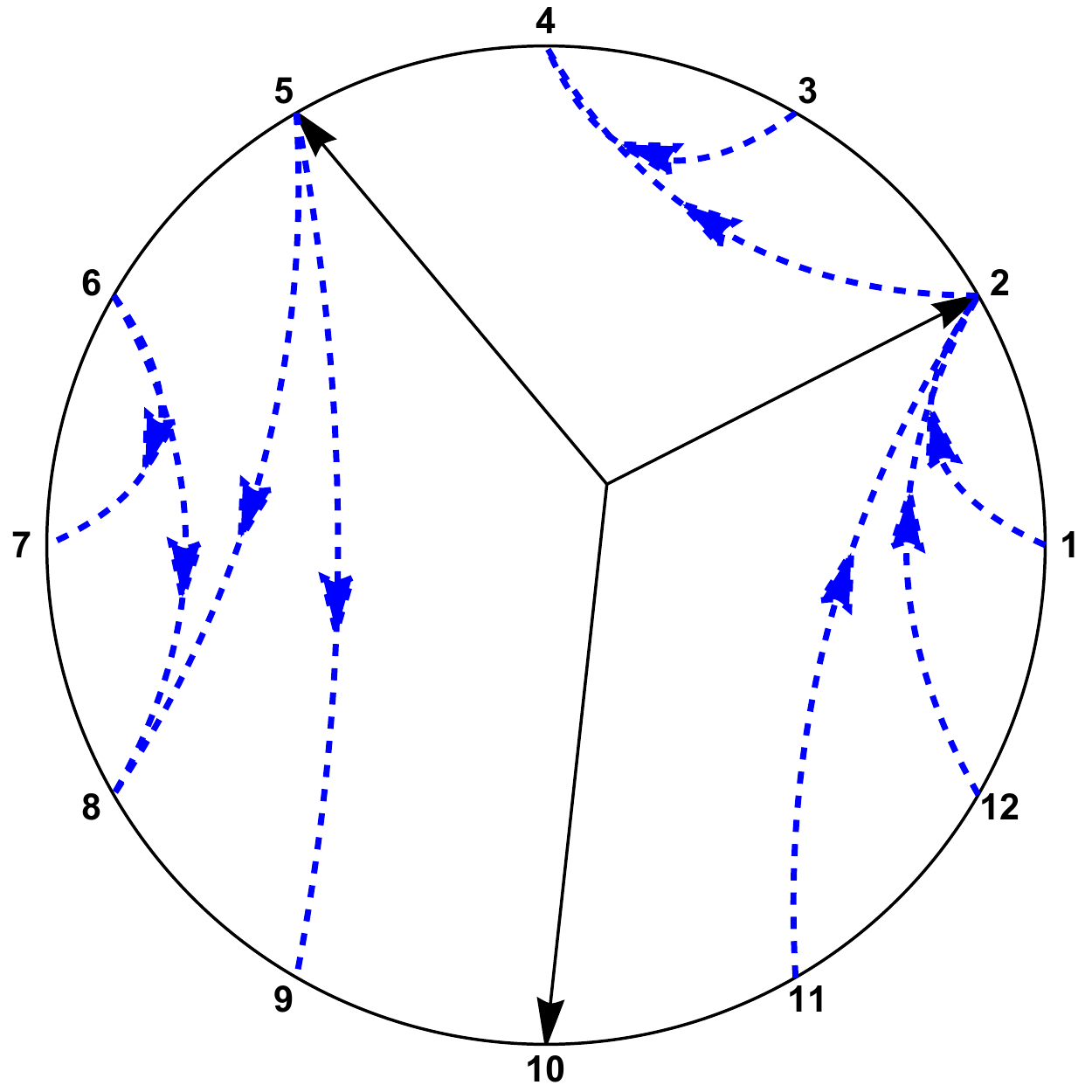}
		\caption{An arc diagram (with a counter-clockwise orientation) for the product of three amputated currents: that is, for the smoothly split $n=10$ particle partial biadjoint scalar amplitude, according to the triple (2,5,10), into a product of three currents.  The (non-amputated) currents would include the three directed arcs $(2 \rightarrow 5),(5 \rightarrow 10)$ and $(10 \rightarrow2)$.  Arcs are directed to indicate that the planar pole should be constructed from the particle labels on the segment of the (counter-clockwise directed) circle that is parallel to the directed arc.}
		\label{fig:feynmandiagramarcs}
	\end{figure}
	
	The diagrammatic scheme as in Figure \ref{fig:feynmandiagramarcs} is intended to have a double meaning.  First, it provides a natural structure which organizes the expansion of the product of the three currents in a smoothly split amplitude, as we discuss below.  Second, it generalizes a well-known construction in combinatorics, from arc diagrams on a line to now arc diagrams on a circle.  It is known that arc diagrams on the line, or in more standard terminology alternating trees, are in bijection with simplices in the Catalan triangulation of the so-called root polytope, first studied in \cite{gelfand1997combinatorics}.  There does not seem to be, in this combinatorial context, a corresponding interpretation of directed arc diagrams in a disk with oriented boundary.  In \cite{Early:2021tce} root polytopes and their triangulations were generalized using certain higher $k$ analogs of simple $A_n$ roots and noncrossing geometries.  See also \cite{Early:2018zuw}, where biadjoint partial amplitudes with one off-shell leg were studied in the context of root polytopes.  The combinatorics of arc diagrams seems reminiscent of the construction in \cite{Gao:2017dek} of so-called Cayley polytopes, in the context of disk integrals.
	
	Indeed, from Figure \ref{fig:feynmandiagramarcs} we can extract one of the contributions to the smoothly split amplitude, restricted to the split kinematic subspace $(2,5,8)$, as the product of three amputated currents,
	\begin{eqnarray*}
		m_{10}(\mathbb{I},\mathbb{I})\big\vert_{\text{split kin.}}  & = & \left(\frac{1}{s_{234}s_{34}} + \cdots \right)\left(\frac{1}{s_{67}s_{678}s_{5678}s_{56789}} + \cdots  \right) \left(\frac{1}{s_{1,2,11,12}s_{12,1,2}s_{1,2}} + \cdots  \right).
	\end{eqnarray*}

	We conclude our discussion of split kinematics with some issues that require further exploration.  We caution that we know very little about the preimage of split kinematics as a subvariety of the Cartesian product of $n$ copies of Minkowski space $\mathbb{R}^{1,D-1}$, as it is the intersection of a large number of hypersurfaces of the form 
	$$p_a\cdot p_b = -x_{a,1}x_{b,1} + \sum_{j=2}^D x_{a,j}x_{b,j}=0,$$
	where 
	$$p_a = (x_{a,1},\ldots, x_{a,D})$$
	for $a=1,\ldots, n$.
	
	Moreover, we do not know in general the minimum dimension $D$ which makes the intersection nontrivial, nor do we know the topology of the subvariety.  Such questions are beyond the scope of this paper and are left to future work.

	\section{Smoothly Splitting Biadjoint Amplitudes}\label{sec: smooth splits biadjoint scalar}
	
	In this section we prove the formula obtained by smoothly splitting biadjoint amplitudes. In many standard quantum field theory arguments Feynman diagrams make properties manifest and they are the standard tool for proofs. However, due to the semi-locality of smooth splits we choose to proceed using the Cachazo-He-Yuan (CHY) formalism. 
	
	In the CHY formalism, partial amplitudes $m_n(\mathbb{I},\mathbb{I})$ are obtained as an integration over the moduli space of $n$ punctures on $\mathbb{CP}^1$, ${\cal M}_{0,n}$, using the scattering equations \cite{Cachazo:2013gna,Cachazo:2013hca,Cachazo:2013iea,Cachazo:2014xea}. 
	
	Consider the following parameterization of ${\cal M}_{0,n}$ using inhomogeneous coordinates for the punctures
	\begin{align}
		\begin{bmatrix}
			\sigma_1 & \sigma_2 & \sigma_3 &  \sigma_4 &  \cdots   & \sigma_n \\
			1 & 1 & 1 & 1 &  \cdots &  1
		\end{bmatrix} / SL(2,\mathbb{C}).
		\label{paramCHY}
	\end{align}

	The CHY potential is defined as a function of the Pl\"ucker coordinates $|a\, b| = \sigma_a-\sigma_b$ and takes the form
	\begin{align}
		\begin{split}
			{\cal S}_n & = \sum_{a<b}s_{ab}\, \textrm{log}|a\, b| 
		\end{split}\,.
		\label{SCHY}
	\end{align}
	It is not difficult to show that ${\cal S}_n$ is invariant under $SL(2,\mathbb{C})$ transformations and therefore one can fix the location of three punctures. 
	
	We are interested in studying the behavior of amplitudes on the $(i,j,k)$ split kinematic subspace. Therefore it is natural to fix $\sigma_i = 0$, $\sigma_j=1$, and $\sigma_k=\infty$. Recall that $1\leq i<j<k\leq n$. 
	
	Here it is convenient to set $i=1$. Note that each term in the potential function \eqref{SCHY} corresponds to an entry in the matrix of Mandelstam invariants shown schematically in Figure \ref{fig:densityplot}. This already shows that the potential function splits into three parts, each corresponding to one of the diagonal blocks in Figure \ref{fig:densityplot} with extra terms corresponding to the rows and columns in the set $\{ 1,j,k\}$. More explicitly, the potential \eqref{SCHY} can be written as
	
	\begin{figure}[h!]
		\centering
		\includegraphics[width=0.45\linewidth]{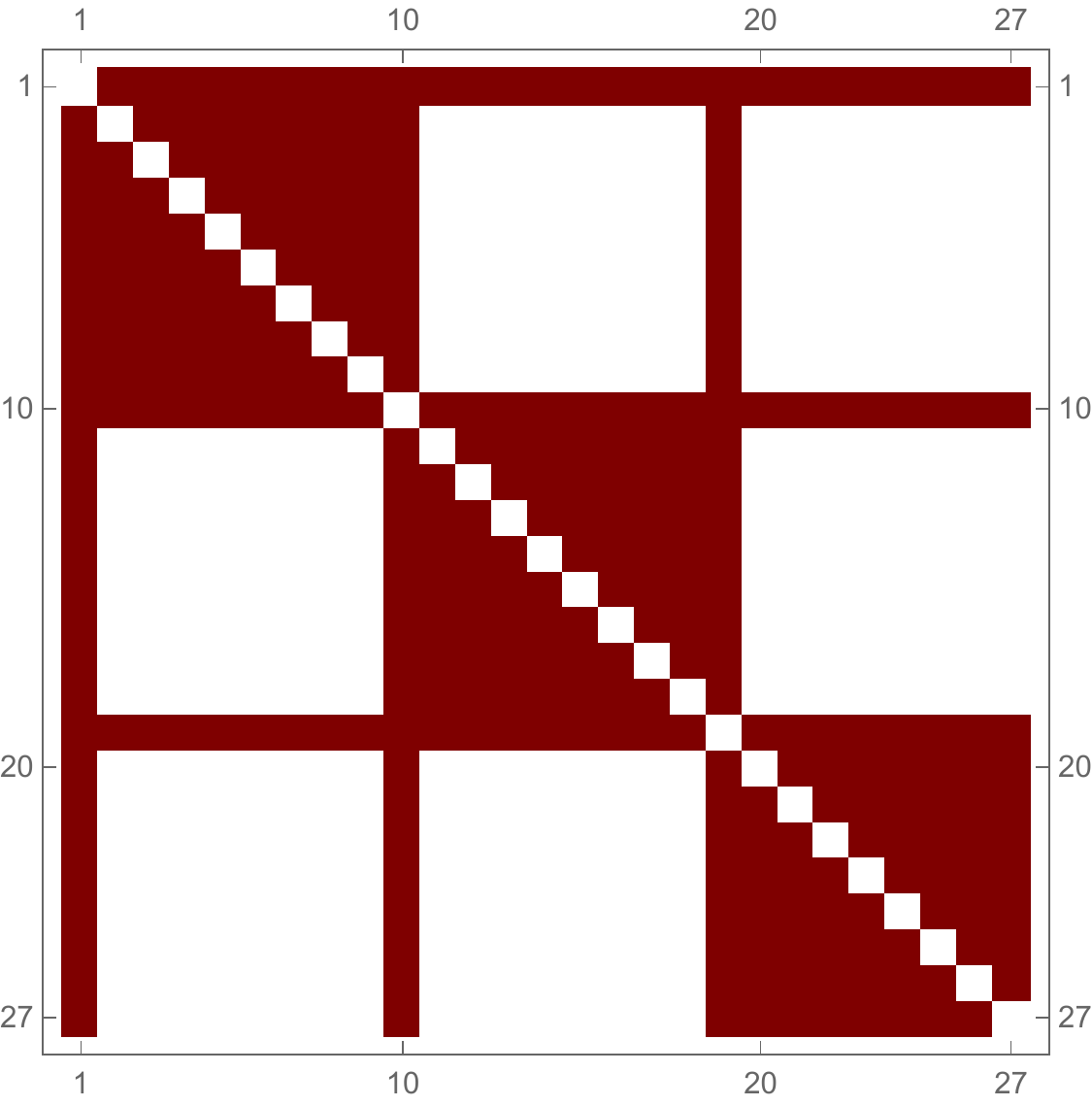}
		\caption{Structure of the matrix of Mandelstam invariants $s_{ij}$ on the split kinematic subspace. For illustration purposes a sufficiently generic case is presented. Here $n=27$ split kinematics $(1,10,19)$ is shown, having set to zero all $s_{ab}$'s in the unshaded regions.}
		\label{fig:densityplot}
	\end{figure}

	\be
	{\cal S}_n = {\cal B}_{(1,j)} + {\cal B}_{(j,k)} + {\cal B}_{(k,1)}+{\cal R}_1 +{\cal R}_j +{\cal R}_k +{\cal T}_{1j}+{\cal T}_{jk}+{\cal T}_{k1}
	\ee
	with the terms coming from the interior of the three blocks
	\be\label{blocksC} 
	{\cal B}_{(1,j)} :=\!\! \sum_{1<a<b<j} \!\! s_{ab}\, \textrm{log}|a\, b|, \quad  {\cal B}_{(j,k)} :=\!\! \sum_{j<a<b<k}\!\! s_{ab}\, \textrm{log}|a\, b|,\quad {\cal B}_{(k,1)} :=\!\! \sum_{k<a<b\leq n}\!\! s_{ab}\, \textrm{log}|a\, b|,
	\ee
	and the extra terms 
	\be\label{extraT}
	{\cal R}_1:=\sum_{a\notin \{1,j,k\}}s_{a1} \textrm{log}|a\, 1|, \quad {\cal R}_j := \sum_{a\notin \{1,j,k\}}s_{aj} \textrm{log}|a\, j|, \quad {\cal R}_k := \sum_{a\notin \{1,j,k\}}s_{ak} \textrm{log}|a\, k|,
	\ee
	\be\label{extraT2}
	{\cal T}_{1j}:=s_{1j} \textrm{log}|1\, j|, \quad {\cal T}_{jk} := s_{jk} \textrm{log}|j\, k|, \quad {\cal T}_{k1} := s_{k1} \textrm{log}|k\, 1|.
	\ee

	Using the $SL(2,\mathbb{C})$ gauge fixing described above, \eqref{extraT} and \eqref{extraT2} become
	\be\label{extraT3}
	{\cal R}_1 =\!\! \sum_{a\notin \{1,j,k\}}s_{a1} \textrm{log}(\sigma_a), \quad {\cal R}_j =\!\! \sum_{a\notin \{1,j,k\}}s_{aj} \textrm{log}(1-\sigma_a), \quad {\cal R}_k+ {\cal T}_{k1} +{\cal T}_{jk} = 0,\quad {\cal T}_{1j} = 0.
	\ee
	In the last two equations we have used momentum conservation in the form $s_{1k}+s_{2k}+\ldots + s_{nk} = 0$ and that $\textrm{log}(\sigma_1-\sigma_j) =\textrm{log}(1) = 0$, respectively. 
	
	The non-vanishing terms in \eqref{extraT3} can be redistributed into the three parts from the blocks in  \eqref{blocksC} to define
	\be\label{blocksF} 
	\begin{split}
		{\cal S}_{(1,j)} := & \!\! \sum_{1<a<b<j} \!\! s_{ab}\, \textrm{log}|a\, b| + \sum_{1<a<j}s_{a1} \textrm{log}(\sigma_a) + \sum_{1<a<j}s_{aj}\textrm{log}(1-\sigma_a), 
		\\  {\cal S}_{(j,k)} := & \!\! \sum_{j<a<b<k}\!\! s_{ab}\, \textrm{log}|a\, b|+ \sum_{j<a<k}s_{a1} \textrm{log}(\sigma_a) + \sum_{j<a<k}s_{aj}\textrm{log}(1-\sigma_a),\\ {\cal S}_{(k,1)} := & \!\! \sum_{k<a<b\leq n}\!\! s_{ab}\, \textrm{log}|a\, b|+ \sum_{k<a\leq n}s_{a1} \textrm{log}(\sigma_a) + \sum_{k<a\leq n}s_{aj}\textrm{log}(1-\sigma_a).
	\end{split}
	\ee
	Using this the CHY potential \eqref{SCHY} can be written as
	\be 
	{\cal S}_n = {\cal S}_{(1,j)} + {\cal S}_{(j,k)} + {\cal S}_{(k,1)}.
	\ee
	Close inspection of this formula reveals something remarkable. Each term only depends on the location of the non-fixed punctures within the range specified by the labels. For example, ${\cal S}_{(j,k)}$ is only a function of $\sigma_{a}$ with $j<a<k$. 
	
	Having analysed the behavior of the CHY potential function, the next step is to study the CHY integral representation of the amplitude. We start by writing the formulation with $\sigma_1,\sigma_j$ and $\sigma_k$ set to generic values,
	\begin{align}\label{chyBA}
		\begin{split}
			m_n(\mathbb{I},\mathbb{I})=\int\prod_{a=1}^n\! ' d\sigma_a\delta\left(\frac{\partial {\cal S}}{\partial \sigma_a}\right)\left(|1\, j||j\, k||k\, 1|\, \textrm{PT}_n(\mathbb{I})\right)^2\,,
		\end{split}
	\end{align}
	where the prime in the product means that $a\notin \{i,j,k\}$. Here PT stands for Parke-Taylor function or factor and it is defined as
	\begin{align}\label{PTdef}
		\begin{split}
			\textrm{PT}_n(\mathbb{I}) := \frac{1}{|1\,2||2\, 3|\cdots |n-1\, n| | n \, 1|}.
		\end{split}
	\end{align}
	Let us write the combination that appears in the integrand of \eqref{chyBA} more explicitly showing the locations of labels $1,j,k$,
	\begin{align}\label{PTA}
		\begin{split}
			|1\, j||j\, k||k\, 1| \textrm{PT}_n(\mathbb{I}) := \frac{|1\, j||j\, k||k\, 1|}{|1\,2||2\, 3|\cdots |j-1,j||j\,j+1|\cdots |k-1\, k||k\, k+1|\cdots  |n-1\, n| | n \, 1|}.
		\end{split}
	\end{align}
	Using the gauge fixing $\sigma_1 = 0$, $\sigma_j=1$, and $\sigma_k=\infty$ one finds that \eqref{PTA} becomes
	\begin{align}\label{factPT}
		\begin{split}
			& \left( \frac{1}{\sigma_2 |2\, 3|\cdots |j-2\,j-1|(\sigma_{j-1}-1)}\right)\times  \left(\frac{1}{(1-\sigma_{j+1})|j+1\, j+2| \cdots |k-2\, k-1|}\right)\times \\ & \left(\frac{1}{|k+1\, k+2|\cdots  |n-1\, n|\sigma_n}\right).
		\end{split}
	\end{align}
	Once again, each of the factors depends only on the variables in one of the three sets defined by the potentials ${\cal S}_{(1,j)}$, ${\cal S}_{(j,k)}$, and ${\cal S}_{(k,1)}$.
	
	Reorganizing the CHY integral \eqref{chyBA} one finds that it splits into three factors, i.e.
	\begin{equation}\label{identity}
		\begin{split}
			m_n(\mathbb{I},\mathbb{I})|_{\rm split\, kin.} = 
			& \left(\int\prod_{a=2}^{j-1}d\sigma_a\delta\left(\frac{\partial {\cal S}_{(1,j)}}{\partial \sigma_a}\right) \textrm{PT}_{(1,j)}\right)\left(\int\prod_{a=j+1}^{k-1}d\sigma_a\delta\left(\frac{\partial {\cal S}_{(j,k)}}{\partial \sigma_a}\right) \textrm{PT}_{(j,k)}\right) \\
			&\left(\int\prod_{a=k+1}^{n}d\sigma_a\delta\left(\frac{\partial {\cal S}_{(k,1)}}{\partial \sigma_a}\right) \textrm{PT}_{(j,k)}\right),
		\end{split}
	\end{equation}
	with $\textrm{PT}_{(1,j)}$, $\textrm{PT}_{(j,k)}$ and $\textrm{PT}_{(k,1)}$ defined as each of the factors in \eqref{factPT} respectively. 
	
	The last step is the identification of each factor in \eqref{identity} with amputated currents.  
	
	In order to complete the argument let us start by reinterpreting the potential functions ${\cal S}_{(1,j)}$, ${\cal S}_{(j,k)}$, and ${\cal S}_{(k,1)}$ in \eqref{blocksF}. 
	The first function ${\cal S}_{(1,j)}$ can be thought of as the CHY potential for a current with one off-shell particle with momentum $P_K:=-p_1-p_2-\ldots -p_j$ and gauge fixed so that $\sigma_1 = 0$, $\sigma_j = 1$ and $\sigma_K=\infty$. Here we follow the definition given by Naculich in \cite{Naculich:2015zha} and reviewed in Appendix \ref{sec: currents}. Note that we have introduced the notation $K$ for the off-shell leg and should not be confused with the $k^{\text{th}}$ particle of the original amplitude. 
	
	The second function ${\cal S}_{(j,k)}$ requires a rearrangement before it can be identified. Note that
	\be 
	\sum_{j<a<k}s_{a1}\log (\sigma_a) = 2\sum_{j<a<k}p_{a}\cdot p_1\log (\sigma_a) .
	\ee
	Using momentum conservation,
	$$p_1 = -(p_2+p_3+\cdots p_{j-1})-(p_j+p_{j+1}+\cdots +p_k)-(p_{k+1}+p_{k+2}+\cdots +p_n)$$
	and noticing that on the $(1,j,k)$ split kinematic subspace 
	$$ p_a\cdot p_1 = -p_a\cdot (p_j+p_{j+1}+\cdots +p_k) \quad \forall \; a: \; j<a<k $$
	once can write ${\cal S}_{(j,k)}$ in \eqref{blocksF} as
	\be 
	{\cal S}_{(j,k)} =\sum_{j<a<b<k}\!\! s_{ab}\, \textrm{log}|a\, b|+ \sum_{j<a<k}2p_{a}\cdot P_I \textrm{log}\, (\sigma_a) + \sum_{j<a<k}s_{aj}\textrm{log}\, (1-\sigma_a).
	\ee 
	Comparing the formula in the appendix \eqref{naculich} it is straightforward to conclude that this is the CHY potential for a current with off-shell momentum $P_I := -p_j-p_{j+1}-\cdots -p_k$ and gauge fixed so that $\sigma_I=0$, $\sigma_j = 1$ and $\sigma_k=\infty$. 
	
	Finally, the function ${\cal S}_{(k,1)}$ in \eqref{blocksF} can be written as
	\be 
	{\cal S}_{(k,1)} =\sum_{k<a<b\leq n}\!\! s_{ab}\, \textrm{log}|a\, b|+ \sum_{k<a\leq n}s_{a1} \textrm{log}\, (\sigma_a) + \sum_{k<a\leq n}2p_a\cdot P_J\textrm{log}\, (1-\sigma_a)\,,
	\ee 
	where $P_J := -p_k-p_{k+1}-\cdots -p_n-p_1$. Comparing to \eqref{naculich} one has a current gauge fixed so that $\sigma_1=0$, $\sigma_J=1$, and $\sigma_k=\infty$.
	
	Let us reinterpret the factors into which the Parke-Taylor function in Equation \eqref{factPT} decomposed, i.e. $\textrm{PT}_{(1,j)}$, $\textrm{PT}_{(j,k)}$ and $\textrm{PT}_{(k,1)}$.   
	
	Consider 
	\be
	\textrm{PT}_{(1,j)} = \left( \frac{1}{\sigma_2 |2\, 3|\cdots |j-2\,j-1|(\sigma_{j-1}-1)}\right).
	\ee
	This is indeed a standard $|1j||jK||K1|\textrm{PT}(1,2,\ldots ,j-1,j,K)$ with the gauge fixing $\sigma_1=0$, $\sigma_j = 1$ and $\sigma_K=\infty$. Likewise, 
	$$ \textrm{PT}_{(j,k)} =\left. |jk||kI||Ij|\textrm{PT}(j,j+1,\ldots ,k-1,k,I)\right|_{\sigma_I=0,\sigma_j=1,\sigma_k= \infty}$$ 
	and 
	$$\textrm{PT}_{(k,1)} =\left. |k1||1J||Jk|\textrm{PT}(k,k+1,\ldots ,n,1,J)\right|_{\sigma_1=0,\sigma_J=1,\sigma_k= \infty}.$$
	Combining all these results the final form of the biadjoint amplitude on the $(1,j,k)$ split kinematic subspace is
	\begin{equation}
		\begin{split}
			m_n(\mathbb{I},\mathbb{I})|_{\rm split\, kin.} = {\cal J}(1,2,...,j){\cal J}(j,j+1,...,k) {\cal J}(k,k+1,...,n,1).
		\end{split}
	\end{equation}
	The three amputated currents were defined in terms of Feynman diagrams in Section \ref{sec: smooth splits intro} and their CHY formulations are discussed in detail in Appendix \ref{sec: currents}.

	\section{Smoothly Splitting NLSM and Special Galileon Amplitudes}\label{sec: smooth splits others}
	
	In this section we derive and study how split kinematics induces smooth splits in two other theories of scalars that admit a CHY formulation: the $U(N)$ non-linear sigma model (NLSM) and the special Galileon.
	
	\subsection{NLSM Amplitudes}
	
	Historically, interest in the NLSM started from studying an effective field theory of interactions of Goldstone bosons known as pions \cite{Gell-Mann:1960mvl}. It is well-known that in this theory, when a single particle becomes soft, scattering amplitudes vanish implying that there must be a non-linearly realized symmetry. This phenomenon is known as the Adler zero \cite{Adler:1965ga,Susskind:1970gf}. Instead, the double soft limit is the relevant one when one tries to obtain information about the spontaneously broken symmetries of the theory \cite{Arkani-Hamed:2008owk}. The lagrangian of the NLSM can be written as \cite{Cachazo:2014xea}
	\begin{equation}
		\begin{split}
			{\cal L}_{\textrm{NLSM}}=\frac{1}{8\lambda^2}\textrm{Tr}(\partial_{\mu}\textrm{U}^{\dagger}\partial^{\mu}\textrm{U})\,,
		\end{split}
		\label{lagNLSM}
	\end{equation}
	where we have used the Cayley transform to write $\textrm{U}=(\mathbb{I}_{N\times N}+\lambda\Phi)(\mathbb{I}_{N\times N}-\lambda\Phi)^{-1}$. Here $\Phi=\phi_IT^I$ where $\phi_I$ are the scalars carrying a flavour index, $T^I$ are the $U(N)$ generators, and $\lambda$ is a constant.
	
	The CHY formula for NLSM amplitudes, which is non-vanishing only for an even number of particles, was proposed in \cite{Cachazo:2014xea} as
	\begin{equation}
		\begin{split}
			A_n^{\textrm{NLSM}}(\mathbb{I})=\int d\mu_n \textrm{PT}_n(\mathbb{I})\,\textrm{det}'\textbf{A}_n\,,
		\end{split}
		\label{nlsmchy}
	\end{equation}
	where we have defined the CHY measure $d\mu_n\equiv\prod_{a=0}^{n-3}d\sigma_a\delta\left(\frac{\partial {\cal S}}{\partial \sigma_a}\right)$ and $\textbf{A}_n$ is an $n\times n$ dimensional matrix with entries $A_{ab}\equiv \frac{s_{ab}}{\sigma_a-\sigma_b}$.
	
	In \eqref{nlsmchy}, $\textrm{det}'\textbf{A}_n$ is the reduced determinant of $\textbf{A}_n$ and is defined as 
	$$\textrm{det}'\textbf{A}_n:=\frac{1}{(\sigma_{p}-\sigma_q)^2}\,\textrm{det}\textbf{A}_n^{[p\,q]},$$
	where $\textbf{A}_n^{[p\,q]}$ is the submatrix of $\textbf{A}_n$ defined by removing the $p^{\rm th}$ and $q^{\rm th}$ rows and columns. This reduction is necessary since the matrix $\textbf{A}_n$ has co-rank 2 on the support of the delta functions in the measure. It is not difficult to show that $\textrm{det}'\textbf{A}_n$ is independent of the choice of $p$ and $q$.
	
	To start the study of the behaviour of NLSM amplitudes under split kinematics, let us first repeat the argument that led to the conclusion that only $d=3$-splits are possible for the biadjoint amplitude presented in \eqref{degreeC}. NLSM amplitudes have degree one in Mandelstam invariants (or equivalently, mass dimension two) for any values of $n$. This immediately implies that it is impossible to smoothly split NLSM amplitudes in terms of NLSM amputated currents which also have the same degree as the amplitudes. This leads to the expection that NLSM amplitudes should vanish on split kinematics. However, considering explicit examples reveals a surprising result. Directly evaluating the $n=8$ NLSM amplitude on split $(1,3,6)$ kinematics gives rise to
	\begin{equation}
		\begin{split}
			A^{\textrm{NLSM}}_8(\mathbb{I})\vert_{(1,3,6)}=\left(s_{12}+s_{23}\right)\left(\frac{s_{34}+s_{45}}{s_{345}}+\frac{s_{45}+s_{56}}{s_{456}}-1\right)\left(\frac{s_{67}+s_{78}}{s_{678}}+\frac{s_{78}+s_{81}}{s_{781}}-1\right)\,.
		\end{split}
		\label{exnlsm8}
	\end{equation}
	The first factor has the form of an $n=4$ NLSM amputated current and hence degree one. The second and third factors in the split do not have the form of NLSM amplitudes. In fact, five-point NLSM amplitudes vanish. These new currents therefore belong to a theory that extends the NLSM and have dimension zero leading to a consistent split. It is surprising that by simply exploring a region of the space of Mandelstam invariants one can find amplitudes of a different theory emerging from those of the original one. Exactly the same phenomenon was observed by Cachazo, Cha, and Mizera \cite{Cachazo:2016njl} when they computed the coefficient of the Adler zero and found exactly the same kind of extended amplitudes. These so-called {\it mixed amplitudes} involve NLSM particles (pions) and biadjoint scalars. 
	
	The particular currents in \eqref{exnlsm8} correspond to mixed 5-point amputated currents of pions and biadjoint scalars \cite{Cachazo:2016njl}, where particles $1$, $3$, $6$ and the new off-shell ones with momenta $-(p_3+p_4+p_5+p_6)$ and $-(p_6+p_7+p_8+p_1)$ are identified with biadjoint scalars, while the rest are NLSM scalars.
	
	In fact, we will show that smoothly splitting NLSM amplitudes either vanish or produce one amputated current of the standard NLSM theory as well as two amputated currents in the extended NLSM theory, i.e. with an odd number of biadjoint scalars and an even number of NLSM scalars.

	The CHY formulation of all mixed NLSM amplitudes corresponding to the extended theory was found in \cite{Cachazo:2016njl}. It contains an additional $U(\tilde{N})$ flavour group and a biadjoint scalar field. Its CHY formula is 
	\begin{equation}
		\begin{split}
			A_n^{\textrm{NLSM}\oplus\phi^3}(\mathbb{I}\vert\beta) =
			\int d\mu_n \textrm{PT}_n(\mathbb{I})\, \textrm{PT}_{\beta}\, \textrm{det}\textbf{A}_{\bar{\beta}}\, .
		\end{split}
		\label{nlsmmix}
	\end{equation}
	The notation here requires some explanation. Both species of particles share the canonical ordering, $\mathbb{I}$, but the biadjoint scalars also respect the ordering $\beta$ in the $U(\tilde{N})$ flavour group indices. Here $\bar{\beta}$ represents the particles in the complement of the set $\beta$ in $[n]$. It is also common in the literature to use $\textrm{PT}_{\beta}\equiv\textrm{PT}(\beta)$ in order to avoid cluttering of the formulas. 
	
	Let us present the general result for 3-splits of NLSM amplitudes postponing the proof to Section \ref{sec: NLSM proof}. 
	
	At first sight there seem to be four cases to consider. As in previous sections, cyclic invariance allows us to study $(1,j,k)$-split kinematics without losing generality. The cases correspond to the different choices for the parity of $j$ and $k$. However, one can check that all choices except for $j\in 2\mathbb{Z}+1$ and $k\in 2\mathbb{Z}+1$, lead to one current with an even number of points and two with an odd number of points. The case $j\in 2\mathbb{Z}+1$ and $k\in 2\mathbb{Z}+1$ requires all three currents to have an even number of points. This, however, is not possible as discussed above and leads to a vanishing result, i.e. a zero of the amplitude, as shown in Section \ref{detproof}.
	
	Let us present the explicit result for $j\in 2\mathbb{Z}+1$ and $k\in 2\mathbb{Z}$, knowing that other cases can be obtained by reflections and relabeling,
	\begin{equation}
		\begin{split}
			&  A_n^{\textrm{NLSM}}(\mathbb{I})|_{\rm split\, kin.}
			= {\cal J}_{j+1}^{\textrm{NLSM}}(\mathbb{I})\times {\cal J}_{k-j+2}^{\textrm{NLSM}\oplus\phi^3}(\mathbb{I}\vert\beta_1)\times {\cal J}_{n-k+3}^{\textrm{NLSM}\oplus\phi^3}(\mathbb{I}\vert\beta_2)\,,
		\end{split}
		\label{NLSMsplit}
	\end{equation}
	Here $\beta_1 =\{ I,j,k \} $, $\beta_2 = \{ 1,J,k \}$, with $I,J$ denoting off-shell legs, and with the currents defined using the CHY formula \eqref{nlsmmix} as explained in more detail in Section \ref{sec: NLSM proof}. A simple argument using degree (or mass dimension) counting reveals that having three biadjoint particles in each mixed current is the only possibility\footnote{The reason why $|\beta_1| = |\beta_2|=3$ is the following. The degree of an amputated mixed current is $3-n+|\bar{\beta}_i|$ and that of a NLSM amputated currents is $1$. Using this in \eqref{NLSMsplit} imposes the constraint 
		$1=j-n+2+|\bar{\beta}_1|+|\bar{\beta}_2|$. Since $|\beta_1|+|\bar{\beta}_1|=k-j+2$ and $|\beta_2|+|\bar{\beta}_2|=n-k+3$, it must be that $|\beta_1|+|\beta_2|=6$. Mixed amplitudes only exist for $|\beta_i|>2$ and therefore $|\beta_1|=|\beta_2|=3$.}.
	
	\subsection{Special Galileon Amplitudes}
	
	The second theory we study in this section is the special Galileon, which was discovered in \cite{Cachazo:2014xea} (see also \cite{Cheung:2014dqa,Hinterbichler:2015pqa}) and whose CHY formula is given by 
	
	\begin{equation}
		\begin{split}
			A_n^{\textrm{sGal}}=\int d\mu_n (\textrm{det}'\textbf{A}_n)^2\,,
		\end{split}
		\label{sgalchy}
	\end{equation}
	where $d\mu_n$ is the same CHY measure used in other theories and $\textrm{det}'\textbf{A}_n$ is the same reduced determinant appearing in the NLSM CHY formula.
	
	This theory is a special case of some scalar theories known as Galileon theories, which have appeared in different contexts, e.g. in cosmology and in the decoupling limit of massive gravity \cite{Hinterbichler:2011tt,Dvali:2000hr,Kampf:2014rka}. The general Galileon lagrangian is given by
	
	\begin{equation}
		\begin{split}
			{\cal L}_{\textrm{Gal}}=-\frac{1}{2}\partial_{\mu}\phi\partial^{\mu}\phi+\sum_{m=3}^{\infty}g_m\hspace{0.5mm}\phi\hspace{0.5mm}\textrm{det}\{\partial^{\mu_ a}\partial_{\nu_b}\phi\}_{a,b=1}^{m-1}\,,
		\end{split}
		\label{lagGal}
	\end{equation}
	which computes non-vanishing amplitudes for any number of particles. However, the special Galileon amplitude \eqref{sgalchy} vanishes for an odd number of particles. It also vanishes when a single particle becomes soft. 
	
	Special Galileon amplitudes have degree $n-1$ in the kinematic invariants, i.e. they have mass dimension $2(n-1)$. Once again the same analysis as done in \eqref{degreeC} reveals that it is impossible to find a smooth splitting of special Galileon amplitudes in terms of special Galileon amputated currents which also have the same degree as the amplitudes. This again leads to the expectation that special Galileon amplitudes should vanish on split kinematics. 
	
	Another reason not to expect a smooth splitting is that, unlike biadjoint scalar and NLSM amplitudes, special Galileon particles do not have any flavour structure and hence no ordering, i.e., the complete permutation invariant amplitude must be considered\footnote{Note that biadjoint and NLSM amplitudes are also permutation invariant since their fields are bosons. However, the flavour structure allows for a decomposition in terms of color-ordered partial amplitudes.}. This implies that it contains a permutation invariant set of poles. This means that the Mandelstam invariants set to zero in a given split kinematics point could be producing singularities in the amplitude. Indeed, this is the case: some individual Feynman diagrams do diverge.  
	
	All this makes it surprising that special Galileon amplitudes smoothly split. Moreover, it is by using its CHY formulation, which re-sums Feynman diagrams, that the behavior on split kinematics is most easily understood. For this reason, we do not need to take a limit to produce smooth splits. Instead, smooth splits appear directly when imposing split kinematics to its CHY formula. 
	
	From the NSLM amplitude discussion it is reasonable to expect that special Galileon amplitudes split into products of mixed amputated currents. We recall the CHY formula for the most general mixed amplitudes, which now involve all three kinds of particles discussed so far \cite{Cachazo:2016njl},
	\begin{equation}
		\begin{split}
			A_n^{\textrm{sGal}\oplus\textrm{NLSM}^2\oplus\phi^3}(\alpha\vert\beta)=\int d\mu_n \left(\textrm{PT}_{\alpha}\,\textrm{det}\textbf{A}_{\bar{\alpha}}\right)\left(\textrm{PT}_{\beta}\, \textrm{det}\textbf{A}_{\bar{\beta}} \right)\, .
		\end{split}
		\label{sgalmixnlsmv2}
	\end{equation}
	This extended theory contains a $U(N)\times U(\tilde{N})$ biadjoint scalar and a NLSM field for each of the two flavour groups. Here the biadjoint scalars correspond to labels $\alpha\cap\beta$ while the special Galileon particles correspond to labels $\bar{\alpha}\cap\bar{\beta}$. The $U(N)$ and $U(\tilde{N})$ NLSM particles correspond to $\alpha\cap\bar{\beta}$ and $\bar{\alpha}\cap\beta$, respectively. 
	
	Once again, the behavior of sGal amplitudes on $(1,j,k)$-split kinematics depends on the parity of $j$ and $k$. The amplitudes vanish when both $j$ and $k$ are odd and splits in terms of an amputated current of the original theory times two mixed currents corresponding to mixed amplitudes of the special form when $\alpha=\beta$ in \eqref{sgalmixnlsmv2}, i.e.
	\begin{equation}
		\begin{split}
			A_n^{\textrm{sGal}\oplus\phi^3}(\beta)=\int d\mu_n \textrm{PT}_{\beta}^2(\textrm{det}\textbf{A}_{\bar{\beta}})^2\,.
		\end{split}
		\label{sgalmixnlsm}
	\end{equation}

	The final formula for $(1,j,k)$-split kinematics with $j\in 2\mathbb{Z}+1$ and $k\in 2\mathbb{Z}$, knowing that other cases can be obtained by reflections and relabeling, is
	\begin{equation}
		\begin{split}
			A_n^{\textrm{sGal}}|_{\textrm{split kin.}}= {\cal J}_{j+1}^{\textrm{sGal}}\times {\cal J}_{k-j+2}^{\textrm{sGal}\oplus\phi^3}(\beta_1)\times {\cal J}_{n-k+3}^{\textrm{sGal}\oplus\phi^3}(\beta_2)\,.
		\end{split}
		\label{splitsgal}
	\end{equation}
	Here $\beta_1 =\{ I,j,k \} $, $\beta_2 = \{ 1,J,k \}$, with $I,J$ denoting off-shell legs. We present the proof of this formula in Section \ref{sec: sGal Proof}.

	\subsection{Behaviour of $\textrm{det}'\textbf{A}_n$ on Split Kinematics}\label{detproof}

	In the following subsections we use the CHY argument seen for the biadjoint scalar theory in Section \ref{sec: smooth splits biadjoint scalar} to derive how 3-splits appear in NLSM and special Galileon theories under split kinematics. In order to achieve it, we first have a look at the behavior of the reduced determinant that enters into the CHY formulation of these theories, under split kinematics.
	
	Recall that the reduced determinant is independent of the choice of removing any two rows and columns. Therefore, we can remove row and column $1$ and we still have to remove one more row and column.
	
	\begin{figure}[H]
		\includegraphics[width=1\linewidth]{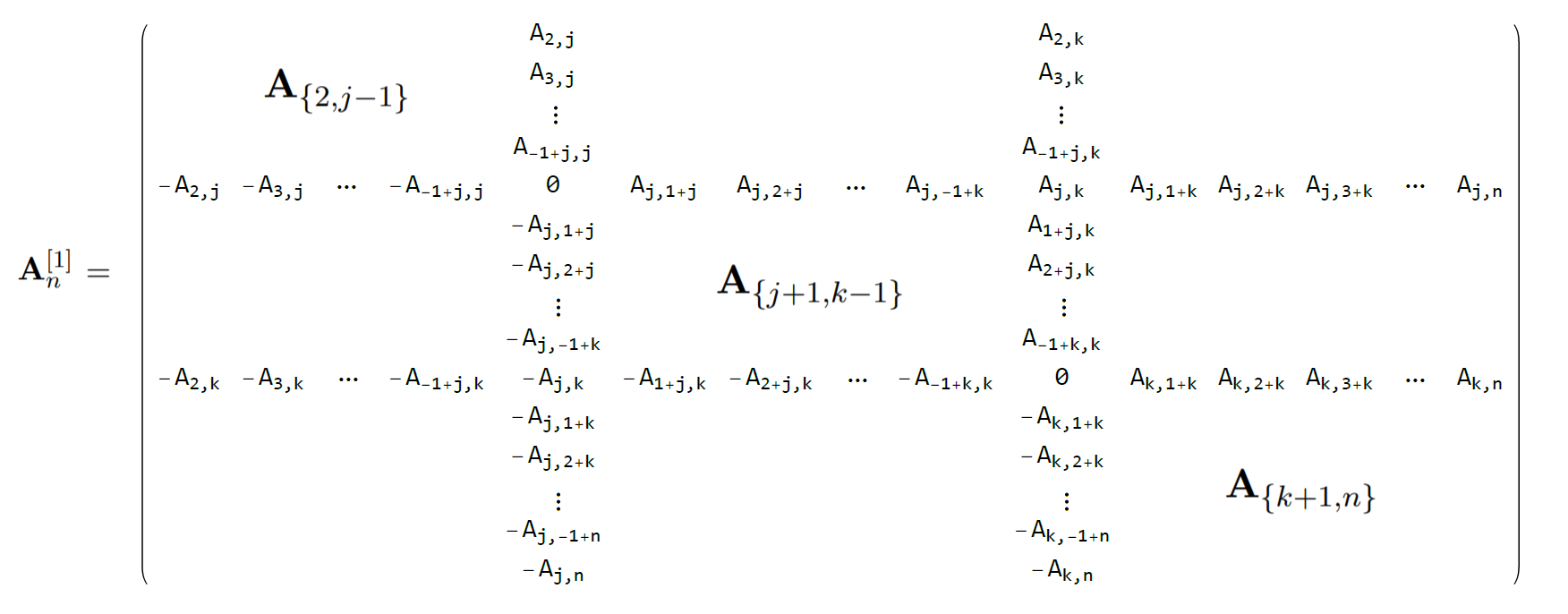}
		\centering
		\caption{General form of the matrix $\textbf{A}_n^{[1]}$, i.e. when we remove row and column 1.}
		\label{fig:matrixA}
	\end{figure}
	
	Without loss of generality, consider again the split kinematics $(1,j,k)$. Under this kinematics, the matrix $\textbf{A}_n$ after removing row and column $1$ has the form of the matrix in Figure \ref{fig:matrixA}, where the entries are $A_{a,b}\equiv\frac{s_{ab}}{\sigma_a-\sigma_b}$ and $\textbf{A}_{\{a,b\}}$ are matrices defined as
	
	\begin{eqnarray}
		\textbf{A}_{\{a,b\}} = \begin{bmatrix}
			0 & A_{a,a+1} & A_{a,a+2} & \cdots & A_{a,b} \\
			-A_{a,a+1} & 0 & A_{a+1,a+2} & \cdots & A_{a+1,b}  \\
			\vdots & \vdots & \ddots & \vdots & \vdots  \\
			-A_{a,b} & -A_{a+1,b} & \cdots & -A_{b-1,b} & 0  \\
		\end{bmatrix}\,.
		\label{matrixSk}
	\end{eqnarray}
	We point out that $\textbf{A}_{\{1,n\}}=\textbf{A}_n$ in this notation. The rest of the entries are zero. For the argument we will use the following lemma:
	\begin{lem}
		\label{detLemma}
		Let $M\in \mathbb{C}^{2m\times 2m}$ be antisymmetric and $L\in \mathbb{C}^{r\times r}$ generic, then  
		\begin{equation}
			\textrm{det}\left[ 
			\begin{array}{c|c} 
				\large{M} & \begin{array}{ccccc}
					0 & 0 & 0 & \cdots & 0  \\
					0 & 0 & 0 & \cdots & 0 \\
					\vdots & \vdots & \vdots & \ddots & \vdots \\
					0 & 0 & 0 & \cdots & 0 \\
					c_1 & c_2 & c_3 & \cdots & c_{r}
				\end{array} \\ 
				\hline 
				\begin{array}{ccccc}
					0 & 0 & \cdots & 0 & d_1  \\
					0 & 0 & \cdots & 0 & d_2 \\
					\vdots & \vdots & \ddots & \vdots & \vdots \\
					0 & 0 & \cdots & 0 & d_r\\
				\end{array} & \large{L}
			\end{array} 
			\right] =\textrm{det}(M)\textrm{det}(L)
			\label{det}
		\end{equation}
		for any values of $d_a$ and $c_a$.
	\end{lem}
	
	The proof of the lemma is very simple and we present it in Appendix \ref{apA}.
	
	Now recall from the CHY proof in the biadjoint scalar that the potential splits into three terms ${\cal S}_{(1,j)}$, ${\cal S}_{(j,k)}$ and ${\cal S}_{(k,1)}$, where ${\cal S}_{(1,j)}$ produces an amputated current with $j+1$ legs, ${\cal S}_{(j,k)}$ produces an amputated current with $k-j+2$ legs and ${\cal S}_{(k,1)}$ produces an amputated current with $n-k+3$ legs. Also recall that for non-vanishing NLSM and special Galileon amplitudes $n$ is always even. 
	
	Let us consider the case in which $j\in 2\mathbb{Z}+1$ and $k\in 2\mathbb{Z}$. Motivated by the fact that in the following subsections we send puncture $\sigma_k$ to infinity, here we remove row and column $k$ from the matrix to end up with the that in Figure \ref{fig:matrixj}, where the upper-left block $\textbf{A}_{\{2,j\}}$ is $(j-1)\times (j-1)$ dimensional, and therefore even-dimensional. We also note that $\textbf{A}_{\{j+1,k-1\}}$ has dimension $k-j-1$ and that $\textbf{A}_{\{k+1,n\}}$ has dimension $n-k$. Given the statement \eqref{det} above, and the fact that the determinant of a block-diagonal matrix is the product of the determinants of each block, we know that $\textrm{det}\textbf{A}_n^{[1k]}=\textrm{det}\textbf{A}_{\{2,j\}}\,\textrm{det}\textbf{A}_{\{j+1,k-1\}}\,\textrm{det}\textbf{A}_{\{k+1,n\}}$.

	\begin{figure}[h!]
		\includegraphics[width=1\linewidth]{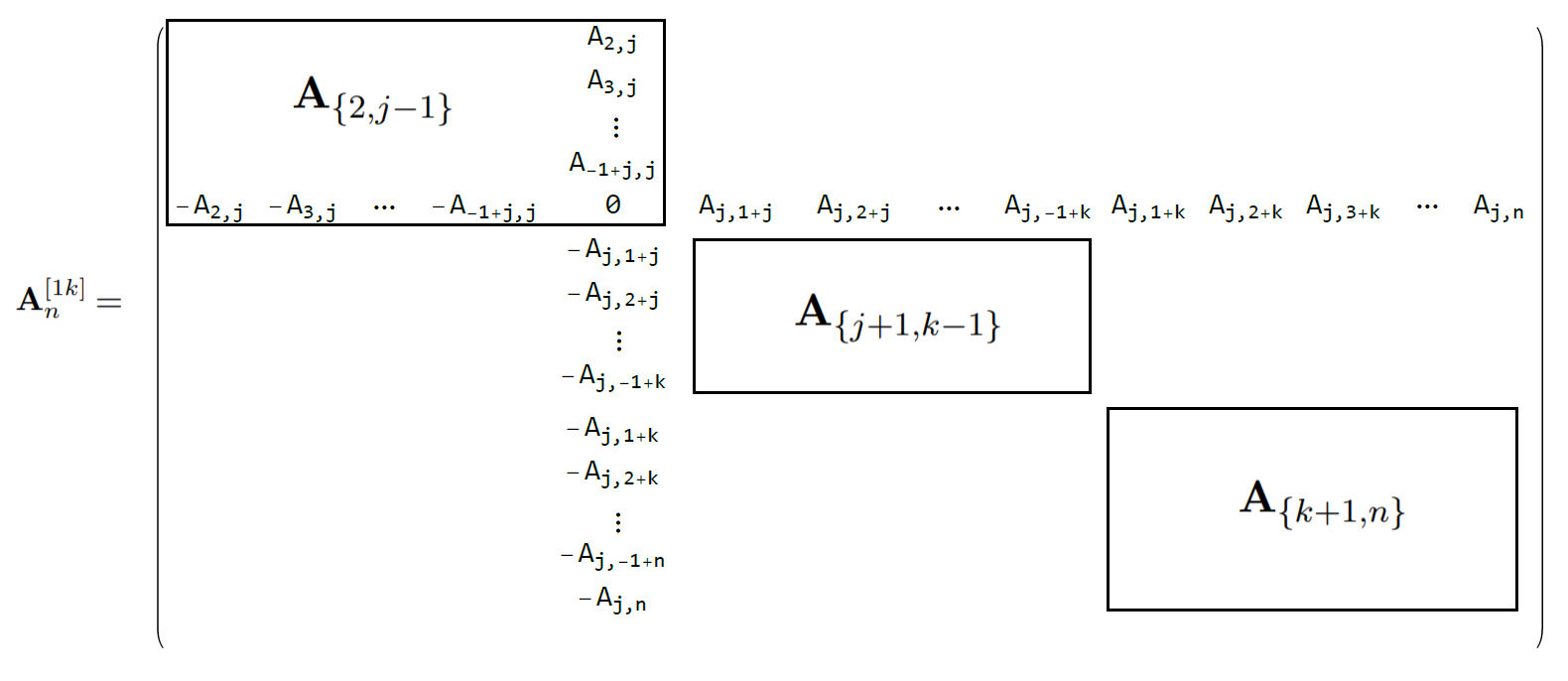}
		\centering
		\caption{General form of the matrix $\textbf{A}_{n}^{[1k]}$ where we emphasize the three different blocks that play the role in its determinant.}
		\label{fig:matrixj}
	\end{figure}
	
	Now notice that if $k$ is even then $\textbf{A}_{\{j+1,k-1\}}$ and $\textbf{A}_{\{k+1,n\}}$ are even-dimensional. In this case the block $\textbf{A}_{\{2,j\}}$ will give rise to the NLSM or special Galileon amputated current, whilst each of the two blocks $\textbf{A}_{\{j+1,k-1\}}$ and $\textbf{A}_{\{k+1,n\}}$ that are embedded into a bigger one will give rise to the mixed amputated currents. 
	
	If $k$ is odd then $\textbf{A}_{\{j+1,k-1\}}$ and $\textbf{A}_{\{k+1,n\}}$ are odd-dimensional and therefore the whole determinant vanishes since the determinant of an odd-dimensional antisymmetric matrix is zero. What this is telling us is the following. When $j$ and $k$ are odd, we know that all of the three amputated currents will have an even number of external particles, since $n$ is even. Hence, the determinant is protecting the whole object from becoming a product of only non-mixed amputated currents!

	\subsection{Proof for NLSM Amplitudes}\label{sec: NLSM proof}
	
	Now we are ready to prove how 3-splits are produced in NLSM amplitudes. Without loss of generality, we consider again the split kinematics $(1,j,k)$. Recall from Section \ref{sec: smooth splits biadjoint scalar} that under this kinematics the CHY potential ${\cal S}_n$ splits into ${\cal S}_{(1,j)}$, ${\cal S}_{(j,k)}$ and ${\cal S}_{(k,1)}$, which are the potentials given by the parameterizations \eqref{matrix1}, \eqref{matrix2} and \eqref{matrix3}, respectively, with their corresponding particle identifications
	
	\begin{equation}
		{\cal S}_{(1,j)}:\begin{blockarray}{cccccccc}
			\textcolor{Maroon}{1} &  \textcolor{Maroon}{2} &  \textcolor{Maroon}{3} &  \textcolor{Maroon}{4} &  &  \textcolor{Maroon}{j-1} & \textcolor{Maroon}{j} &  \textcolor{Maroon}{K} \\
			\begin{block}{[cccccccc]}
				0 & \sigma_2 & \sigma_3 & \sigma_4 & \cdots & \sigma_{j-1} & 1 & 1 \\
				1 & 1 & 1 & 1 & \cdots & 1 & 1 & 0   \\
			\end{block}
		\end{blockarray}\,,
		\label{matrix1}
	\end{equation}
	
	\begin{equation}
		{\cal S}_{(j,k)}:\begin{blockarray}{ccccccc}
			\textcolor{Maroon}{I} &  \textcolor{Maroon}{j} &  \textcolor{Maroon}{j+1} &  \textcolor{Maroon}{j+2} &  &  \textcolor{Maroon}{k-1} & \textcolor{Maroon}{k} \\
			\begin{block}{[ccccccc]}
				0 & 1 & \sigma_{j+1} & \sigma_{j+2} & \cdots & \sigma_{k-1} & 1  \\
				1 & 1 & 1 & 1 & \cdots & 1 & 0  \\
			\end{block}
		\end{blockarray}
		\label{matrix2}
	\end{equation}
	and
	
	\begin{equation}
		{\cal S}_{(k,1)}: \begin{blockarray}{ccccccc}
			\textcolor{Maroon}{J} &  \textcolor{Maroon}{k} &  \textcolor{Maroon}{k+1} &
			\textcolor{Maroon}{k+2} & &  \textcolor{Maroon}{n} &
			\textcolor{Maroon}{1}\\
			\begin{block}{[ccccccc]}
				1 & 1 & \sigma_{k+1} & \sigma_{k+2} & \cdots & \sigma_n & 0  \\
				1 & 0 & 1 & 1 & \cdots & 1 & 1   \\
			\end{block}
		\end{blockarray}\,.
		\label{matrix3}
	\end{equation}
	Let us consider again $j\in 2\mathbb{Z}+1$ and $k\in 2\mathbb{Z}$, without loss of generality.
	From Section \ref{detproof} we know that the determinant also splits like $\textrm{det}\textbf{A}_n^{[1k]}=\textrm{det}\textbf{A}_{\{2,j\}}\,\textrm{det}\textbf{A}_{\{j+1,k-1\}}\,\textrm{det}\textbf{A}_{\{k+1,n\}}$.
	
	Given the above separation of the moduli space and that of $\textrm{det}\textbf{A}_n^{[1k]}$, one can identify every factor in the smooth split with an amputated current. Namely, we will note that ${\cal S}_{(1,j)}$ generates an amputated current with an even number of particles. The other two factors, given by ${\cal S}_{(j,k)}$ and ${\cal S}_{(k,1)}$, will correspond to amputated currents with an odd number of particles, and therefore are mixed amputated currents.
	
	Let us see this in more detail. Before going to the split kinematics subspace, and after gauge fixing the punctures $\sigma_1=0$, $\sigma_j=1$ and $\sigma_k=\infty$, the NLSM CHY formula \eqref{nlsmchy} picks up two copies of the Fadeev-Popov factor $|1j||jk||k1|$ and is expressed as
	
	$$A_n^{\textrm{NLSM}}(\mathbb{I})=\int d\mu_n(|1j||jk||k1|)^2\textrm{PT}_n(\mathbb{I})\textrm{det}'\textbf{A}_n^{[1k]}\,.$$ One copy of the Fadeev-Popov factor cancels with $\textrm{det}'\textbf{A}_n^{[1k]}=\frac{\textrm{det}\textbf{A}_n^{[1k]}}{(\sigma_1-\sigma_k)^2}$ to produce a finite object. The second copy combines with the Parke-Taylor to produce the neat separation shown in Equation \eqref{factPT}, given by the product of Parke-Taylors $\textrm{PT}_{(1,j)}\times\textrm{PT}_{(j,k)}\times\textrm{PT}_{(k,1)}$ defined in Section \ref{sec: smooth splits biadjoint scalar}. 
	
	Now we go to the split kinematics subspace $(1,j,k)$ with $j\in 2\mathbb{Z}+1$ and $k\in 2\mathbb{Z}$. Recall that in this kinematics the determinant of the original matrix splits as $\textrm{det}\textbf{A}_n^{[1k]}=\textrm{det}\textbf{A}_{\{2,j\}}\,\textrm{det}\textbf{A}_{\{j+1,k-1\}}\,\textrm{det}\textbf{A}_{\{k+1,n\}}$. This implies that in this subspace the NLSM amplitude separates into three pieces
	
	$$\left(\int d\mu_{(1,j)}\textrm{PT}_{(1,j)}\textrm{det}\textbf{A}_{\{2,j\}}\right)\left(\int d\mu_{(j,k)}\textrm{PT}_{(j,k)}\textrm{det}\textbf{A}_{\{j+1,k-1\}}\right)\left(\int d\mu_{(k,1)}\textrm{PT}_{(k,1)}\textrm{det}\textbf{A}_{\{k+1,n\}}\right)$$ where $d\mu_{(a,b)}$ is the CHY measure defined by ${\cal S}_{(a,b)}$. Notice that any dependence on $\sigma_k$ has disappeared.
	
	Let us first analyze the first factor in detail. From \eqref{matrix1} and the definition of the reduced determinant we know that $\textrm{det}'\textbf{A}_{{\cal S}_{(1,j)}}=\frac{\textrm{det}\textbf{A}_{\{2,j\}}}{(\sigma_1-\sigma_K)^2}$ where $\textbf{A}_{{\cal S}_{(1,j)}}$ is the matrix with elements $\frac{s_{ab}}{\sigma_{a}-\sigma_b}$ generated by \eqref{matrix1}. We also note that 
	
	$$\textrm{PT}_{(1,j)}=|1j||jK||K1|\,\textrm{PT}(12\cdots jK)\vert_{\sigma_1=0,\,\sigma_j=1,\,\sigma_K=\infty}\,.$$ This implies that if we start with the expression 
	
	$$\int d\mu_{(1,j)}(|1j||jK||K1|)^2\,\textrm{PT}(12\cdots jK)\textrm{det}'\textbf{A}_{{\cal S}_{(1,j)}}$$ as we had gauge-fixed punctures $\sigma_1$, $\sigma_j$ and $\sigma_K$, then $\frac{1}{(\sigma_1-\sigma_K)^2}$ is what is needed to combine with one copy of the Fadeev-Popov factor $|1j||jK||K1|$ to make the expression finite when $\sigma_K=\infty$, which becomes
	
	$$\int d\mu_{(1,j)}\textrm{PT}_{(1,j)}\textrm{det}\textbf{A}_{\{2,j\}}\,.$$
	Additionally, from \eqref{matrix2} one can see that if the set $\beta_1=\{I,j,k\}$ is identified with the biadjoints, where the complement is given by $\bar{\beta}_1=\{j+1,...,k-1\}$, then we have $\textrm{det}\textbf{A}_{\bar{\beta}_1}=\textrm{det}\textbf{A}_{\{j+1,k-1\}}$. Similarly, from \eqref{matrix3}, if the set $\beta_2=\{1,J,k\}$ is identified with the biadjoints, whose complement is given by $\bar{\beta}_2=\{k+1,...,n\}$, we have $\textrm{det}\textbf{A}_{\bar{\beta}_2}=\textrm{det}\textbf{A}_{\{k+1,n\}}$.
	
	Hence, we see from \eqref{nlsmchy} and \eqref{nlsmmix} that we end up with the 3-split
	\begin{equation}
		\begin{split}
			A_n^{\textrm{NLSM}}(\mathbb{I})|_{\rm split\, kin.}\! =\! & \overbrace{\left(\int d\mu_{(1,j)}(|1j||jK||K1|)^2\,\textrm{PT}(12\cdots jK) \textrm{det}'\textbf{A}_{{\cal S}_{(1,j)}}\right)}^{{\cal J}_{j+1}^{\textrm{NLSM}}(\mathbb{I})}\\
			& \times\underbrace{\left(\int d\mu_{(j,k)}\textrm{PT}_{(j,k)} \textrm{PT}_{\beta_1}\,\textrm{det}\textbf{A}_{\bar{\beta}_1}\right)}_{{\cal J}_{k-j+2}^{\textrm{NLSM}\oplus\phi^3}(\mathbb{I}\vert\beta_1)}\times\underbrace{\left(\int d\mu_{(k,1)}\textrm{PT}_{(k,1)} \textrm{PT}_{\beta_2}\,\textrm{det}\textbf{A}_{\bar{\beta}_2}\right)}_{{\cal J}_{n-k+3}^{\textrm{NLSM}\oplus\phi^3}(\mathbb{I}\vert\beta_2)}\,.
		\end{split}
		\label{516}
	\end{equation} 
	To conclude, we note that, given that the only particles we identify with the biadjoint scalars are contained in the set $\{1,j,k,I,J,K\}$, since every current will contain three of these particles, it follows that we will always have 3 biadjoints in the mixed amputated currents. In fact, the only particle in this set which is not identified with a biadjoint scalar corresponds to the off-shell particle in the non-mixed current. This implies that the non-mixed current will contain two biadjoints and therefore its expression is equivalent to that of a current with only pions.
	
	\subsection{Proof for Special Galileon Amplitudes}\label{sec: sGal Proof}
	
	In this subsection we show that special Galileon amplitudes smoothly split under split kinematics. We make use of the fact that special Galileon amplitudes admit a CHY formulation to derive this behaviour in a similar fashion as with the NLSM amplitudes.
	
	Let us consider again the case with $j\in 2\mathbb{Z}+1$ and $k\in 2\mathbb{Z}$ without loss of generality and recall the separation of moduli spaces given in \eqref{matrix1}, \eqref{matrix2} and \eqref{matrix3}. Also, from Section \ref{detproof} we know that the determinant also splits like $\textrm{det}\textbf{A}_n^{[1k]}=\textrm{det}\textbf{A}_{\{2,j\}}\,\textrm{det}\textbf{A}_{\{j+1,k-1\}}\,\textrm{det}\textbf{A}_{\{k+1,n\}}$. For the same reason as in the previous section, we identify again $\textrm{det}'\textbf{A}_{{\cal S}_{(1,j)}}=\frac{\textrm{det}\textbf{A}_{\{2,j\}}}{(\sigma_1-\sigma_K)^2}$. Also, given that $\beta_1=\{I,j,k\}$ and $\beta_2=\{1,J,k\}$, we can identify the determinants  $\textrm{det}\textbf{A}_{\bar{\beta}_1}=\textrm{det}\textbf{A}_{\{j+1,k-1\}}$ and $\textrm{det}\textbf{A}_{\bar{\beta}_2}=\textrm{det}\textbf{A}_{\{k+1,n\}}$. A similar analysis as in Section \ref{sec: NLSM proof} leads to
	\begin{equation}
		\begin{split}
			A_n^{\textrm{sGal}}|_{\textrm{split kin.}}
			= &  \overbrace{\left(\int d\mu_{(1,j)} (|1j||jK||K1|\,\textrm{det}'\textbf{A}_{{\cal S}_{(1,j)}})^2\right)}^{{\cal J}_{j+1}^{\textrm{sGal}}}\times \overbrace{\left(\int d\mu_{(j,k)} \textrm{PT}_{\beta_1}^2(\textrm{det}\textbf{A}_{\bar{\beta}_1})^2\right)}^{{\cal J}_{k-j+2}^{\textrm{sGal}\oplus\phi^3}(\beta_1)}\\
			& \times\underbrace{\left(\int d\mu_{(k,1)} \textrm{PT}_{\beta_2}^2(\textrm{det}\textbf{A}_{\bar{\beta}_2})^2\right)}_{{\cal J}_{n-k+3}^{\textrm{sGal}\oplus\phi^3}(\beta_2)}
		\end{split}
	\end{equation}
	where we stress again that the fact that the only particles we identify with $\beta_1$ and $\beta_2$ are contained in the set $\{1,j,k,I,J,K\}$ shows why we will always have 3 biadjoints in the mixed amputated currents. Again, the only particle in this set which is not identified with a biadjoint scalar corresponds to the off-shell particle in the non-mixed current. This implies that the non-mixed current will contain two biadjoints and therefore its expression is equivalent to that of a current with only Galileons.

	\section{Applications: New Recursion Relations for NLSM Amplitudes}\label{sec: recursions}
	
	In this section we show how to use smooth splittings of NLSM amplitudes as data to build BCFW-like recursion relations. It is well-known that standard BCFW relations are not applicable to the NLSM. In order to explain the reason let us review the procedure. Consider some subset of momenta and introduce a one-complex parameter deformation $p_i(z) = p_i+z r_i$ such that $p_i(z)^2=0$ and momentum conservation remains valid for all $z$. This means that the amplitude evaluated on this new kinematics can be considered a function $A^{\rm NLSM}(z)$ such that $A^{\rm NLSM}(0)= A^{\rm NLSM}_n(\mathbb{I})$, i.e. it agrees with the desired amplitude at $z=0$. Now 
	$$ A^{\rm NLSM}_n(\mathbb{I}) = \frac{1}{2\pi i}\oint_{|z|=\epsilon} dz \frac{A^{\rm NLSM}(z)}{z}.$$
	Deforming the contour one gets a formula for $A^{\rm NLSM}_n(\mathbb{I})$ in terms of residues where propagators give simple poles. These residues are determined via unitarity to be products of smaller amplitudes and hence the recursive structure. However, there is also the contribution of a pole at $z=\infty$ which is in general not known.
	
	Thus, the condition for the recursion to work is that $A^{\rm NLSM}(z)$ vanishes as $z\to \infty$. In general this is not the case in the NLSM due to the presence of contact terms. One possible solution is to design deformations such that the kinematics becomes that of a soft-limit for some $z=z^*$. Let us choose $z^*=1$. The NLSM is known to vanish in a soft-limit and therefore one can consider 
	$$ A^{\rm NLSM}_n(\mathbb{I}) = \frac{1}{2\pi i}\oint_{|z|=\epsilon} dz \frac{A^{\rm NLSM}(z)}{z(1-z)}.$$
	Now, if the new deformation does not make the behaviour of $A^{\rm NLSM}(z)$ worse as $z\to \infty$ then $\frac{A^{\rm NLSM}(z)}{z(1-z)}$ has a better behavior while its residue at $z=1$ vanishes. As it turns out, either a combination of several of these improvements are needed \cite{Cheung:2015ota} or knowing the behavior of subleading terms in soft limits is needed so that $\frac{A^{\rm NLSM}(z)}{z(1-z)^2}$ can be used \cite{Cachazo:2016njl}. Either way, new information is needed in order to construct a successful recursion relation.
	
	The strategy we will use is therefore to introduce a complex deformation such that at some values $z=z^*$ split kinematics is achieved so that its behaviour can be used instead of soft limits. 
	
	Given that split kinematics is completely defined in terms of Mandelstam invariants, it is convenient to introduce a version of the BCFW procedure for $s_{ab}$ directly without starting with momentum vectors. In general, given a matrix of Mandelstam invariants, a BCFW deformation is achieved by 
	\be\label{bcfw} 
	s_{ab}(z) = s_{ab} + z r_{ab}.
	\ee
	Imposing that $s_{ab}(z)$ is a valid matrix of Mandelstam invariants for any $z$ simply implies that so must be $r_{ab}$. In a sense, \eqref{bcfw} interpolates between two sets of Mandelstam invariants, the original one at $z=0$ and the new one at $z=\infty$. 
	
	In order to construct the desired deformation let us select a particular 3-split $(i,j,k)$. This is achieved by imposing that a certain subset of kinematic invariants vanish. Let us denote such set ${\cal V}_{(i,j,k)}$. For example, for $n=6$ and $(1,3,5)$ one has ${\cal V}_{(1,3,5)} = \{ s_{24},s_{46},s_{62}\}$. Requiring the deformed kinematics to reach the 3-split kinematics at $z=1$ can be achieved by choosing $r_{ab} = -s_{ab}$ if $s_{ab}\in {\cal V}_{(i,j,k)}$. More explicitly, one finds
	\be\label{splitD}
	s_{ab}(z) = \left\{ \begin{array}{cc}
		(1-z)s_{ab} & {\rm if}\quad s_{ab}\in {\cal V}_{(i,j,k)}  \\
		s_{ab}+z r_{ab}  & {\rm otherwise}. 
	\end{array}  \right.
	\ee
	as discussed above, one must require that momentum conservation is satisfied and this means that 
	$$ \sum_{b=1}^n r_{ab} = 0 \quad {\rm for} \quad a\in \{ 1,2,\ldots , n\} . $$
	Let us consider the NLSM amplitude under the deformation \eqref{splitD}. Using the CHY formulation it is easy to show that the mass dimension of $A^{\rm NLSM}_n(\mathbb{I})$ is $2$, i.e. it is of degree one in Mandelstam invariants. This gives 
	\be
	A^{\rm NLSM}(z) = {\cal O}(z)\quad {\rm as} \quad z\to \infty .
	\ee
	This behavior implies that even the modified function $A^{\rm NLSM}(z)/z(1-z)$ still has a pole at $z=\infty$.
	
	The way to solve this problem is to change the deformation so that in addition to reaching $(i,j,k)$-split kinematics at $z=1$ it reaches a different split kinematics point, say $(r,p,q)$, at a different point, say $z=-1$.  
	
	A straightforward way of doing this is by using 
	\be\label{splitD2}
	s_{ab}(z) = \left\{ \begin{array}{cc}
		(1-z)s_{ab} & \quad {\rm if}\quad s_{ab}\in {\cal V}_{(i,j,k)} \quad {\rm and} \quad s_{ab}\notin {\cal V}_{(r,p,q)}  \\
		(1+z)s_{ab} & \quad {\rm if}\quad s_{ab}\notin {\cal V}_{(i,j,k)} \quad {\rm and} \quad s_{ab}\in {\cal V}_{(r,p,q)} \\
		(1-z)(1+z)s_{ab} & \quad {\rm if}\quad s_{ab}\in {\cal V}_{(i,j,k)} \quad {\rm and} \quad s_{ab}\in {\cal V}_{(r,p,q)} \\
		s_{ab}+z r_{ab}  & {\rm otherwise}. 
	\end{array}  \right.
	\ee
	However, this has the problem of making every Mandelstam invariant $s_{ab} \in {\cal V}_{(i,j,k)}\cap  {\cal V}_{(r,p,q)}$ a polynomial of degree $2$ in $z$. Such polynomials would spoil the counting and the construction. 
	
	Therefore we must require that ${\cal V}_{(i,j,k)}\cap  {\cal V}_{(r,p,q)} = \emptyset$.
	
	A simple choice that achieves the desired deformation is $(1,2,4)$ and  $(1,3,4)$. It is easy to prove that ${\cal V}_{(1,2,4)}\cap  {\cal V}_{(1,3,4)} = \emptyset$. More explicitly,
	\begin{align*}\label{twoSets}
		{\cal V}_{(1,2,4)} & =  \{ s_{3a}\; : \; a=5,6,\ldots ,n \},\\
		{\cal V}_{(1,3,4)} & =  \{ s_{2a}\; : \; a=5,6,\ldots ,n \}.
	\end{align*}
	Now we are ready to present the BCFW-like construction. Consider the complex deformation:
	\be\label{splitD3}
	s_{ab}(z) = \left\{ \begin{array}{cc}
		(1-z)s_{2b} & \quad {\rm if}\quad a=2,\;\; b\in \{ 5,6,\ldots ,n \} \\
		(1+z)s_{3b} & \quad {\rm if}\quad a=3,\;\; b\in \{ 5,6,\ldots ,n \}\\
		s_{ab}+z r_{ab}  & {\rm otherwise}. 
	\end{array}  \right.
	\ee
	Here we are using momentum conservation
	\be
	\sum_{b=1}^n s_{ab}(z) = 0. 
	\ee

	The function $A^{\rm NLSM}(z)$ has poles at finite values of $z$ exactly where planar kinematic invariants involving an odd number of particles vanish. This is because the theory possesses interactions vertices with only an even number of legs. Let us call the set of planar invariants in poles 
	\be
	{\cal P}_n :=\{ s_{i,i+1,\ldots ,i+m-1}\; : i\in [n],\; m\in 2\mathbb{Z}\} .
	\ee
	The choice of $r_{ab}$ in \eqref{splitD3} is arbitrary as long as all invariants in ${\cal P}_n$ become polynomials of degree exactly one under the deformation \eqref{splitD3}.   
	
	The BCFW-like formula for the NLSM is then obtained by deforming the contour of 
	$$ A^{\rm NLSM}_n(\mathbb{I}) = \frac{1}{2\pi i}\oint_{|z|=\epsilon} dz \frac{A^{\rm NLSM}(z)}{z(1-z^2)},$$
	giving rise to\footnote{Clearly the original contour $|z|=\epsilon$ is defined to be counterclockwise. The contour deformation leads to contours around the poles at $z=1$, $z=-1$, etc., which are clockwise and therefore the residues pick up an extra minus sign. Also, for contours $|s(z)|=\epsilon$, note that the pole in the amplitude is of the form $1/s(z) \equiv 1/(s+a z)$ for some $a$. This means that the residue of $1/z(1-z^2)(s+az)$ is $-1/(1-(z^*)^2)s$. The minus sign cancels the one needed to make the contour counterclockwise.}
	\be\label{finalBCFW}
	A^{\rm NLSM}_n(\mathbb{I}) = \frac{1}{2}A^{\rm NLSM}(1)+\frac{1}{2}A^{\rm NLSM}(-1) + \sum_{s(z^*)=0 \, :\, s \in {\cal P}_n} A^{\rm NLSM}_L(z^*)\frac{1}{(1-(z^*)^2)s}A^{\rm NLSM}_R(z^*).
	\ee

	In this formula 
	\be
	A^{\rm NLSM}(1) = \left. A^{\rm NLSM}_n(\mathbb{I})\right|_{{\rm split\; kin.\,}(1,3,4)} = {\cal J}^{{\rm NLSM}\oplus \phi^3}(5,\ldots ,n|1,I(1),4)\times \left(s_{13}+r_{13}\right),
	\ee
	where ${\cal J}^{{\rm NLSM}\oplus \phi^3}(5,\ldots ,n|1,I(1),4)$ stands for the $(n-1)$-point current evaluated on $s_{ab}(1)$. 
	
	Likewise, 
	\be
	A^{\rm NLSM}(-1) = \left. A^{\rm NLSM}_n(\mathbb{I})\right|_{{\rm split\; kin.\,}(1,2,4)} = {\cal J}^{{\rm NLSM}\oplus \phi^3}(5,\ldots ,n|1,I(-1),4)\left(s_{24}+r_{24}\right).  
	\ee
	Finally, $A^{\rm NLSM}_L(z^*)$ and $A^{\rm NLSM}_R(z^*)$ are the amplitudes that result from the standard factorization at the planar poles of the deformed amplitude.

	\subsection{Example: Six-Point NLSM Amplitude}
	
	In order to illustrate the BCFW formula \eqref{finalBCFW} let us apply it to the six-point NLSM amplitude. The complex deformation is given by

	\be\label{splitSix}
	s_{ab}(z) = \left\{ \begin{array}{cc}
		(1-z)s_{2b} & \quad {\rm if}\quad a=2,\;\; b\in \{ 5,6 \} \\
		(1+z)s_{3b} & \quad {\rm if}\quad a=3,\;\; b\in \{ 5,6\}\\
		s_{ab}+z r_{ab}  & {\rm otherwise}. 
	\end{array}  \right.
	\ee
	Momentum conservation only imposes six constrains and we find that the remaining freedom can be used to make the following choice
	\begin{equation}
		\begin{split}
			\left\{ \right. & \left. r_{12}\to 0,\, r_{13}\to 0,\, r_{14}\to 0,\, r_{15}\to
			0,r_{16}\to 0,\right.   \\ & \left.
			r_{24}\to -s_{25}-s_{26}-\Lambda^2,\, r_{34}\to s_{35}+s_{36}-\Lambda^2, \right.  \\ & \left.
			r_{45}\to
			s_{26}-s_{36}+\Lambda^2,\, r_{46}\to
			s_{25}-s_{35}+\Lambda^2, \right.
			\\  & \left. r_{56}\to
			-s_{25}-s_{26}+s_{35}+s_{36}-\Lambda^2,\, r_{23}\to \Lambda^2 \right\}.
		\end{split}
	\end{equation}
	Recall that the original $s_{ab}$ are assumed to satisfy momentum conservation. In order to use the recursion formula \eqref{finalBCFW} it is convenient to introduce the planar invariants $s_{123}(z), s_{234}(z)$ and $s_{345}(z)$. These are deformations of the usual planar invariants, e.g.
	\be
	s_{123}(z)=s_{12}(z)+s_{13}(z)+s_{23}(z) = s_{123}+z\Lambda^2.
	\ee
	Now we list the contribution from each of the poles in ${\cal P}_6$. 
	
	The first contribution is from $s_{123}(z^*)=0$, i.e. $z^* = -s_{123}/\Lambda^2$. This is given by
	\be
	\left( s_{12}(z^*)+s_{23}(z^*)\right) \frac{1}{s_{123}(1-(z^*)^2)}
	\left( s_{45}(z^*)+s_{56}(z^*)\right).
	\ee
	The other two contributions are similar. Instead of presenting their expressions as functions of $\Lambda$ we use the fact that the final answer must be $\Lambda$ independent and then present their limit as $\Lambda\to \infty$. In the order $s_{123}(z^*)=0, s_{234}(z^*)=0$ and $s_{345}(z^*)=0$ the contributions are:
	\be
	\begin{split}
		& \frac{(s_{12}+s_{23}-s_{123})(s_{45}+s_{56})}{s_{123}}+  \frac{(s_{23}+s_{34})(s_{56}+s_{61}-s_{234})}{s_{234}}+\\ & \frac{(s_{34}+s_{45}-s_{345})(s_{61}+s_{12}-s_{345})}{s_{345}}.
	\end{split}
	\ee
	Finally, the contributions from split kinematic points $z=1$ and $z=-1$ are computed using mixed currents in the $\textrm{NLSM}\oplus \phi^3$ theory defined in \cite{Cachazo:2016njl},
	\be
	{\cal J}(4,5|1,I(z),3)=\frac{s_{34}(z)+s_{45}(z)}{s_{345}(z)}+\frac{s_{45}(z)+s_{51}(z)}{s_{451}(z)}-1.
	\ee
	This means that 
	\begin{align}
		A^{\rm NLSM}(1) & = {\cal J}(5,6|1,I(1),3)\times (s_{23}(1)+s_{34}(1)), \\
		A^{\rm NLSM}(-1) & = {\cal J}(5,6|1,I(-1),3)\times (s_{12}(-1)+s_{23}(-1)).
	\end{align}
	We also present these results in the limit $\Lambda \to \infty$,
	\be
	\frac{1}{2}A^{\rm NLSM}(1) = 0, \quad \frac{1}{2}A^{\rm NLSM}(-1) = s_{34}+s_{45}-s_{345}.
	\ee
	
	Adding all five contributions gives the expression
	\be
	\begin{split}
		A^{\rm NLSM}_6(\mathbb{I}) = & \frac{(s_{12}+s_{23}-s_{123})(s_{45}+s_{56})}{s_{123}}+  \frac{(s_{23}+s_{34})(s_{56}+s_{61}-s_{234})}{s_{234}}+\\ & \frac{(s_{34}+s_{45}-s_{345})(s_{61}+s_{12}-s_{345})}{s_{345}}+s_{34}+s_{45}-s_{345},
	\end{split}
	\ee
	which agrees with the well-known result
	\be
	A^{\rm NLSM}_6(\mathbb{I}) =\left( \frac{1}{2}\frac{(s_{12}+s_{23})(s_{45}+s_{56})}{s_{123}} - s_{12}\right) + {\rm perm.}
	\ee
	where the permutations indicate five other terms obtained from the one shown by sending all labels $i\to i+m\, {\rm mod}\, 6$ with $m\in \{1,2,3,4,5\}$.

	\section{Discussion}\label{sec:discussion}
	
	In this work we have uncovered a new behavior of tree-level scattering amplitudes on subspaces of the kinematic space. Smoothly splitting amplitudes on the $(i,j,k)$ split kinematic subspace leads to a product of three amputated currents in which the particle set does not partition. This semi-locality is what makes smooth splits different from standard factorization and as far as we know not derivable from unitarity arguments. In fact, the closest behaviour in the literature to smoothly splitting an amplitude seems to be the soft limit.
	
	Obtaining new information on the behaviour of amplitudes on subspaces of kinematic space is important in order to understand what makes such functions special and relevant to the physical world. The semi-local behaviour we have found involves currents which have to be turned into amplitudes in order to be observables. It is interesting to note that when further conditions on the kinematic space are imposed in order to require the currents to become amplitudes at least one of them vanishes. It would be interesting to further explore this phenomenon and perhaps associated with a mechanism for ensuring locality in observables.    
	
	In this paper we have only scratched the surface of this fascinating topic and therefore there are many directions to be explored. Here we only provide a partial list.

	\subsection{Comparison with the Soft Limit}
	
	As mentioned above, the closest behavior to semi-locality seems to be the soft limit. It is therefore instructive to consider the similarities and differences. In a soft limit the momentum of a particle, say the $n^{\rm th}$ particle, is taken to zero, i.e. $p_n\to \tau \hat p_n$ with $\tau\to 0$. In this limit
	\be
	m_n(\mathbb{I}_n,\mathbb{I}_n)\to \left(\frac{1}{s_{n-1,n}}+\frac{1}{s_{n,1}}\right) m_n(\mathbb{I}_{n-1},\mathbb{I}_{n-1}) + {\cal O}(\tau^0).
	\ee
	The so-called soft factor is reminiscent of a four-particle amplitude. Of course, we have seen in this work, this expectation is not correct since $s_{n-1,n,1}/s_{n-1,n} \neq 0$, i.e. the momentum of the fourth leg is off-shell. The ratio is needed in order to remove the trivial $\tau$ dependence. Nevertheless, this soft factor can be thought of as an amputated current ${\cal J}(n-1,n,1)$ and once again we get a semi-local factorization
	\be
	m_n(\mathbb{I}_n,\mathbb{I}_n)\to {\cal J}(n-1,n,1) m_n(\mathbb{I}_{n-1},\mathbb{I}_{n-1}) + {\cal O}(\tau^0)\,,
	\ee
	in which particles $n-1$ and $1$ participate in both factors. 
	
	While the semi-local feature is similar to that of smooth splits the main difference is that this is achieved in a singular limit and there are subleading corrections. 
	
	In order to compare let us consider the $(1,n-2,n-1)$ split kinematic subspace. This is simply defined as the subspace with $s_{an}=0$ for $a\in \{ 2,3,\ldots ,n-3\}$. Here the biadjoint amplitude smoothly splits as
	\be  
	\left. m_n(\mathbb{I}_n,\mathbb{I}_n)\right|_{\rm split~kin.} = {\cal J}(n-1,n,1) {\cal J}(1,2,\ldots ,n-2).
	\ee  
	Note that in order to reach the soft limit subspace from the $(1,n-2,n-1)$ split kinematic subspace one has to impose the additional constrains $s_{n-2,n}=s_{n-1,n}=s_{n,1} = {\cal O}(\tau)$ with $\tau\to 0$. In this limit the off-shell leg of ${\cal J}(1,2,\ldots ,n-2)$ which has momentum $P_I=p_{n-1}+p_n$ becomes $P_I \to p_{n-1}$ and therefore on-shell, turning the current into the amplitude $m_n(\mathbb{I}_{n-1},\mathbb{I}_{n-1})$. It is also worth noticing that the direction in which the soft limit subspace is approached is important. If we were to take the limit $s_{n-2,n}\to 0$ first, then $s_{n-1,n}+s_{n,1} \to 0$ due to momentum conservation and therefore the current ${\cal J}(n-1,n,1)$ would vanish. 
	
	We interpret this close connection between soft limits and how an amplitude smoothly splits as saying that the $(1,n-2,n-1)$ split kinematic subspace provides a ``pre-soft limit". It would be interesting to explore this connection further. 
	
	\subsection{Smoothly Splitting Currents}\label{sec: smoothly splitting currents}
	
	In the standard factorization, which is obtained by expanding around a point where a kinematic invariant vanishes, an amplitude factors as the product of two lower point amplitudes. This means that it is possible to further factor each of the smaller amplitudes. The chain of factorizations leads to the notion of compatible poles or singularities and to many interesting relations to areas of mathematics such as tropical geometry, matroid theory, etc. 
	
	When an amplitude is smoothly split, the resulting factors are amputated currents. This means that it is not obvious that the procedure can be iterated. However, the CHY formulation of the currents seems to suggest that developing a procedure to smoothly split currents is possible and it would be interesting to explore this direction and try to develop all the connections to mathematics which are known for standard factorizations.
	
	Let us consider what happens when we try to split a current and then pose a concrete question for future work.
	
	By restricting $m_8(\mathbb{I},\mathbb{I})$ to the intersection of two split kinematic subspaces, as in Figure \ref{fig:triangulated-4-split-135-157},
	\begin{figure}[h!]
		\centering
		\includegraphics[width=0.5\linewidth]{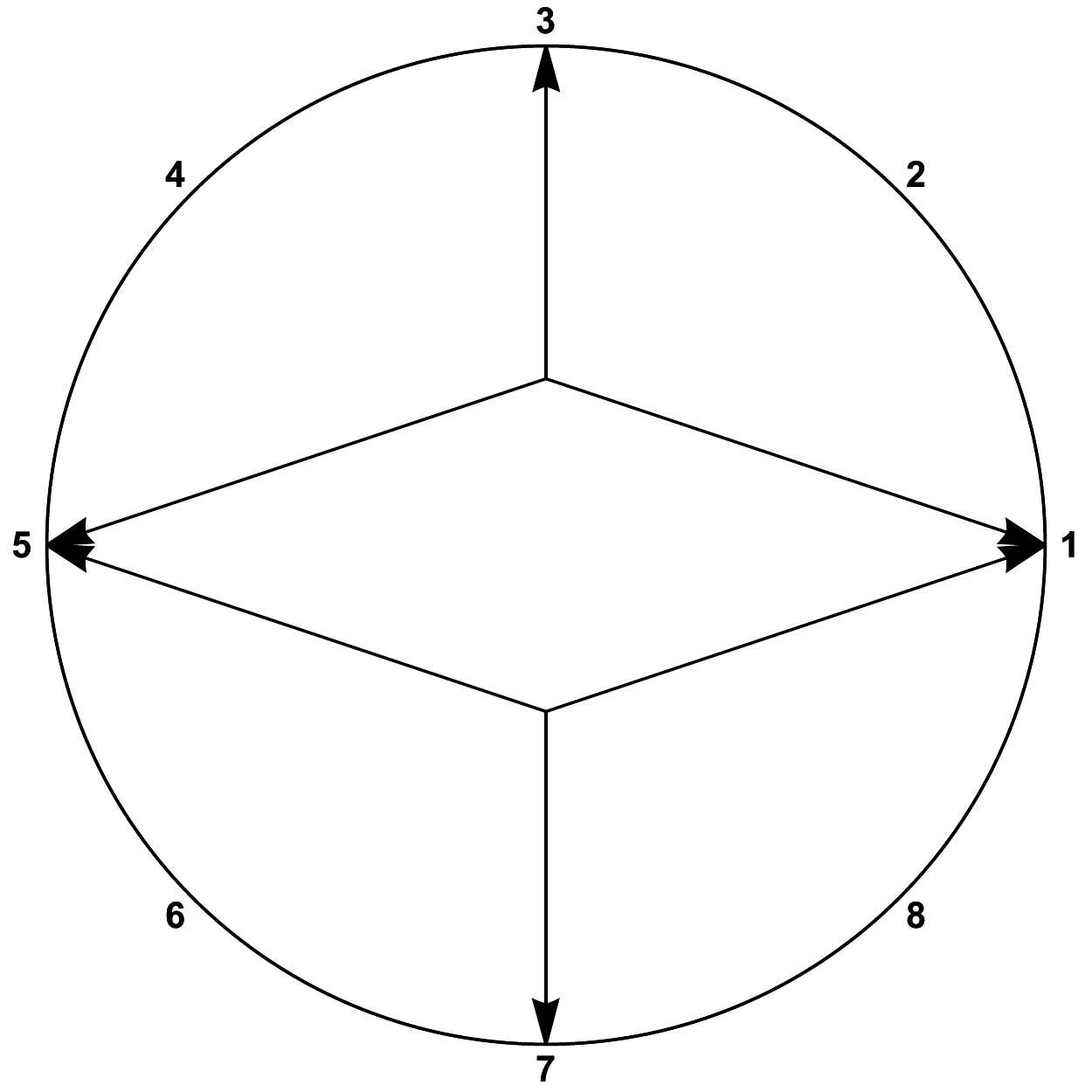}
		\caption{Smoothly splitting a current: intersecting the split kinematics subspaces, $\mathcal{K}_{(1,3,5)}\cap \mathcal{K}_{(1,5,7)}$.}
		\label{fig:triangulated-4-split-135-157}
	\end{figure}
	$\mathcal{K}' = \mathcal{K}_{(1,3,5)} \cap \mathcal{K}_{(1,5,7)}$, we obtain a decomposition of $m_8(\mathbb{I},\mathbb{I})$ into a product of four currents and an additional factor which requires interpretation, 
	$$m_{8} \big\vert_{\mathcal{K}'} = \left(\frac{1}{s_{12}}+\frac{1}{s_{23}}\right)
	\left(\frac{1}{s_{34}}+\frac{1}{s_{45}}\right) \left(\frac{1}{s_{56}}+\frac{1}{s_{67}}\right) \left(\frac{1}{s_{78}}+\frac{1}{s_{18}}\right) \left(\frac{1}{s_{3456}}+\frac{1}{s_{4567}}\right).$$
	Here the fifth factor can be interpreted as correlation function with more than one off-shell leg, but something else stands out more prominently: the first four are linked in a cyclic chain!  We have checked examples at larger $n$, by triangulating polygons, thereby intersecting more split kinematic subspaces to find similar cyclic chain splittings with more than four factors appear for more particles n.  
	
	If it holds in general, what combinatorial structure could govern such decompositions?  Could there be a systematic interpretation of this behavior in terms of subdivisions of the hypersimplex $\Delta_{2,n}$, and/or in terms of associahedra?   One possible approach might be to look towards  (possibly not matroidal) subdivisions of $\Delta_{2,n}$; in any case, the decomposition arising from smoothly splitting currents in this way begs for a combinatorial interpretation.  We present another perspective in the next section.
	
	To be more precise, can one give a complete enumeration and interpretation of all kinematic subspaces $\mathcal{K}'$ such that the restriction $m_n(\mathbb{I},\mathbb{I})\big\vert_{\mathcal{K}'}$ is proportional to the product of a cyclic chain of currents,
	$$m_n(\mathbb{I},\mathbb{I})\big\vert_{\mathcal{K}'} \sim \prod_{j=1}^{d} \mathcal{J}(\alpha_j,\alpha_{j}+1,\ldots,\alpha_{j+1})?$$
	In this case, the additional proportionality factors, which are not shown here, would be a product of correlation functions, each having more than one off-shell leg.
	
	What is the correct setting to explain the cyclic chain decompositions?  Are there cohomological and combinatorial interpretations?  We leave such mathematical investigations and their physical implications to future work.

	\subsection{$\!\!$Additional Comments on Planarity: Kinematic Invariants and Smooth Splits}
	In the usual presentation of Mandelstam invariants $s_{ij}$, they may be organized in a square matrix such as Equation \eqref{eqn: matkinInvariants158}, for which a linear order $1,\ldots, n$ has been chosen, regardless of the symmetry of the scattering process.  In order to represent fully the structure of the kinematic space it is natural to associate them to the vertices of a certain $n-1$ dimensional polytope $\Delta_{2,n}$, with vertices $e_i+e_j$, which is itself permutation invariant.  Moreover, this construction has the convenient feature that the sum of $s_{ij}$ over any facet of $\Delta_{2,n}$ evaluates to zero, due to momentum conservation.  The polytope $\Delta_{2,n}$ is called a hypersimplex and is well-known in combinatorial geometry; see Figure \ref{fig:octahedronmandelstams} for the fundamental example.
	\begin{figure}[h!]
		\centering
		\includegraphics[width=0.5\linewidth]{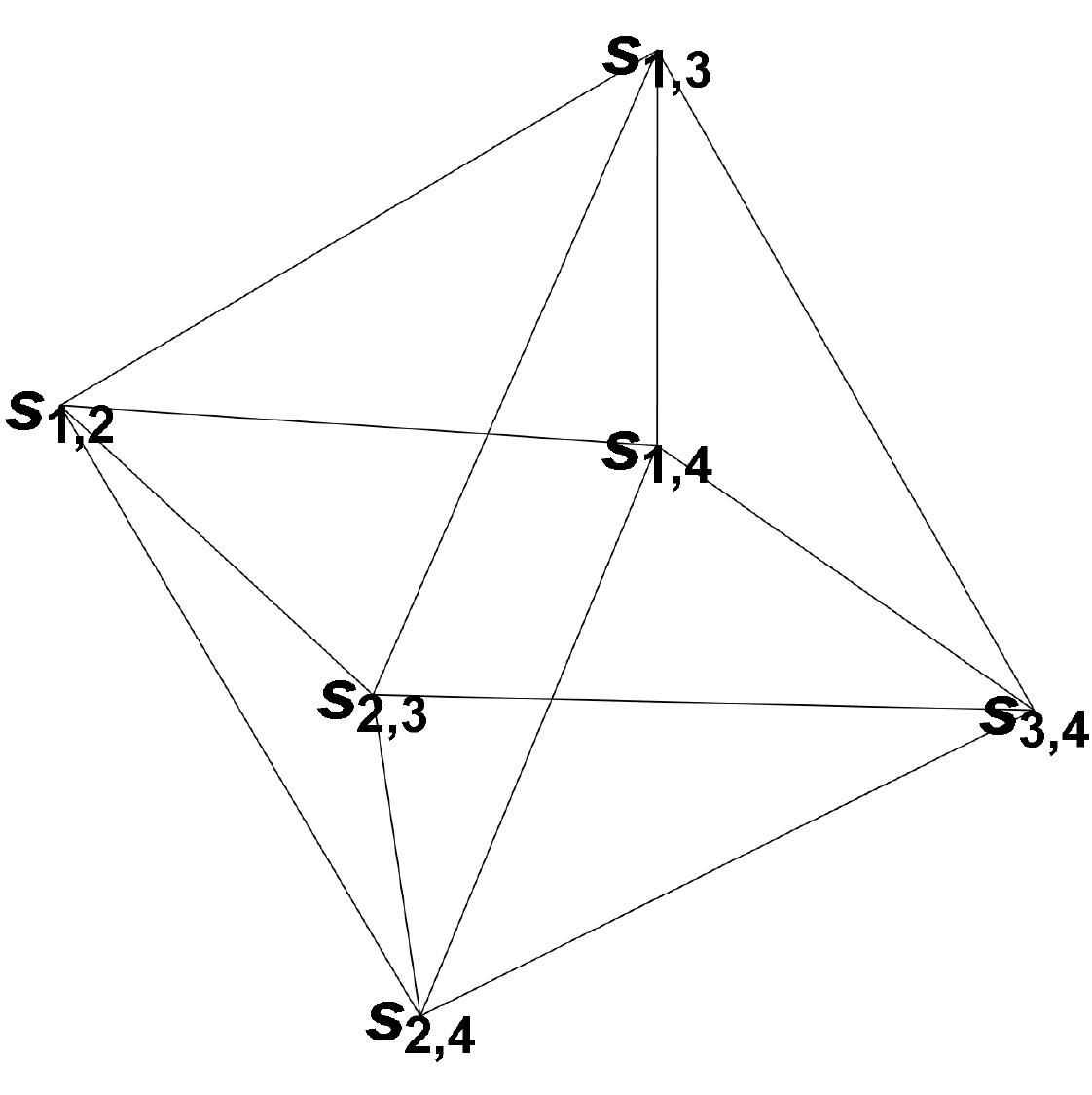}
		\caption{Mandelstam invariants arranged on the hypersimplex $\Delta_{2,4}$, the octahedron.  By momentum conservation, every (triangular) facet sums to zero.}
		\label{fig:octahedronmandelstams}
	\end{figure}
	In this section, we show how this perspective extends to split kinematics; for each split we define a projection of $\Delta_{2,n}$ into $\mathbb{R}^3$; using the projection to identify preimages of lattice points partitions the vertices of $\Delta_{2,n}$ into 16 blocks, as shown in Figure \ref{fig:connectedcomponentshypersimplexsplit158}.
	
	For a nontrivial example of split kinematics, Equation \eqref{eqn: matkinInvariants158} lists the nonzero Mandelstam invariants on the $n=9$ particle split kinematics $(1,5,8)$.
	\begin{eqnarray}\label{eqn: matkinInvariants158}
		(s_{ab}) = \begin{bmatrix}
			0 & s_{12} & s_{13} & s_{14} & s_{15} & s_{16} & s_{17} & s_{18} & s_{19} \\
			s_{12} & 0 & s_{23} & s_{24} & s_{25} & 0 & 0 & s_{28} & 0 \\
			s_{13} & s_{23} & 0 & s_{34} & s_{35} & 0 & 0 & s_{38} & 0 \\
			s_{14} & s_{24} & s_{34} & 0 & s_{45} & 0 & 0 & s_{48} & 0 \\
			s_{15} & s_{25} & s_{35} & s_{45} & 0 & s_{56} & s_{57} & s_{58} & s_{59} \\
			s_{16} & 0 & 0 & 0 & s_{56} & 0 & s_{67} & s_{68} & 0 \\
			s_{17} & 0 & 0 & 0 & s_{57} & s_{67} & 0 & s_{78} & 0 \\
			s_{18} & s_{28} & s_{38} & s_{48} & s_{58} & s_{68} & s_{78} & 0 & s_{89} \\
			s_{19} & 0 & 0 & 0 & s_{59} & 0 & 0 & s_{89} & 0 \\
		\end{bmatrix}
	\end{eqnarray}
	\begin{figure}[h!]
		\centering
		\includegraphics[width=.7\linewidth]{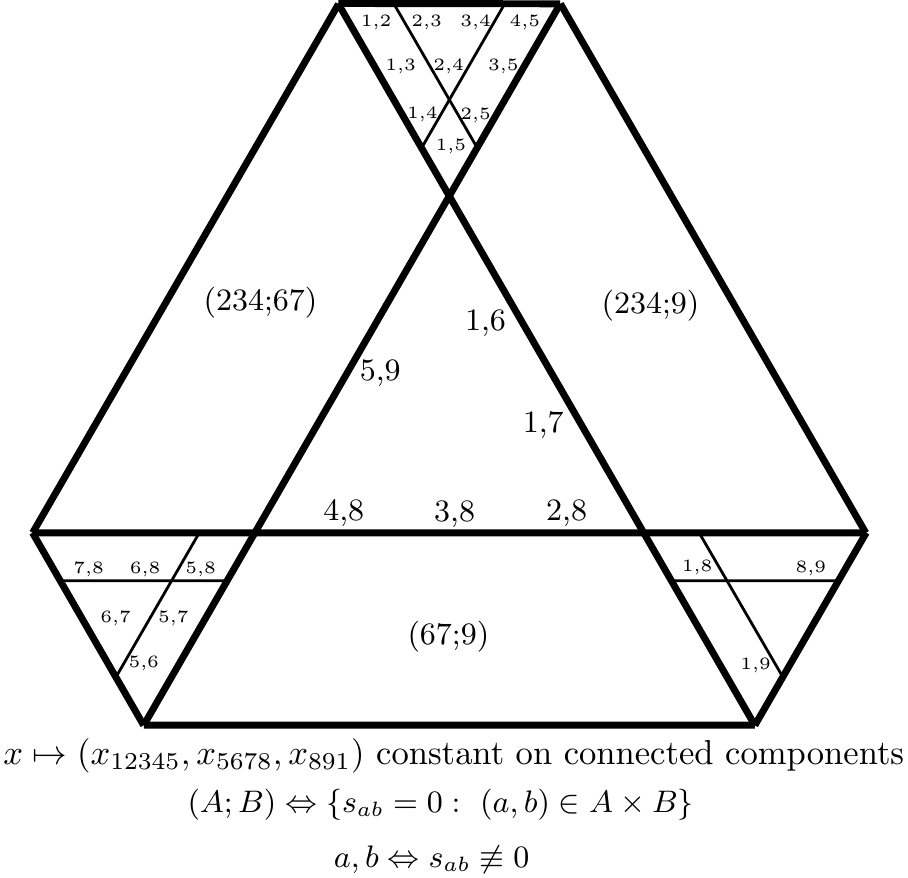}
		\caption{Schematic partition of the vertices $e_a + e_b$ (abbreviated by pairs $a,b$ above) of the hypersimplex $\Delta_{2,9}$ by the split kinematics $(1,5,8)$.  The projection 
			$x \mapsto (x_{12345},x_{5678},x_{891})$ from $\mathbb{R}^9$ to $\mathbb{R}^3$ is by construction constant on connected components.  The Mandelstam invariants which are set to zero are indicated in the three large rhombi.}
		\label{fig:connectedcomponentshypersimplexsplit158}
	\end{figure}
	Figure \ref{fig:connectedcomponentshypersimplexsplit158} organizes the combinatorial data into the 16 connected components of the vertices of the hypersimplex\footnote{Recall that the hypersimplex $\Delta_{2,n}$ is an $n-1$ dimensional polytope, the convex hull of all 0/1 vectors $e_i+e_j \in \mathbb{R}^n$.} $\Delta_{2,9}$, as partitioned by their values under the projection 
	\begin{eqnarray}\label{eq:projection}
		x & \mapsto & (x_1+x_2+x_3+x_4+x_5,x_5+x_6+x_7 +x_8,x_8+x_9 +x_1).
	\end{eqnarray}
	Then it is easy to see that the $n=9$ particle split kinematics $(1,5,8)$ can be extracted from Figure \ref{fig:splitkinematics3dprojection} by specifying which Mandelstams $s_{ab}$ are set to zero: they fit inside the three large rhombi, where for instance $(67;9)$ means that $s_{69}=0$ and $s_{79}=0$.
	The three sets of Mandelstam invariants in the corner triangles play an obvious role for constituent three amputated currents in the 3-split amplitude; however the presence of the six Mandelstams 
	$$s_{16},s_{17},s_{28}, s_{38},s_{48},s_{59}$$
	in the inner triangle is more subtle and is essential in the proof of the splitting.
	\begin{figure}[h!]
		\centering
		\includegraphics[width=0.7\linewidth]{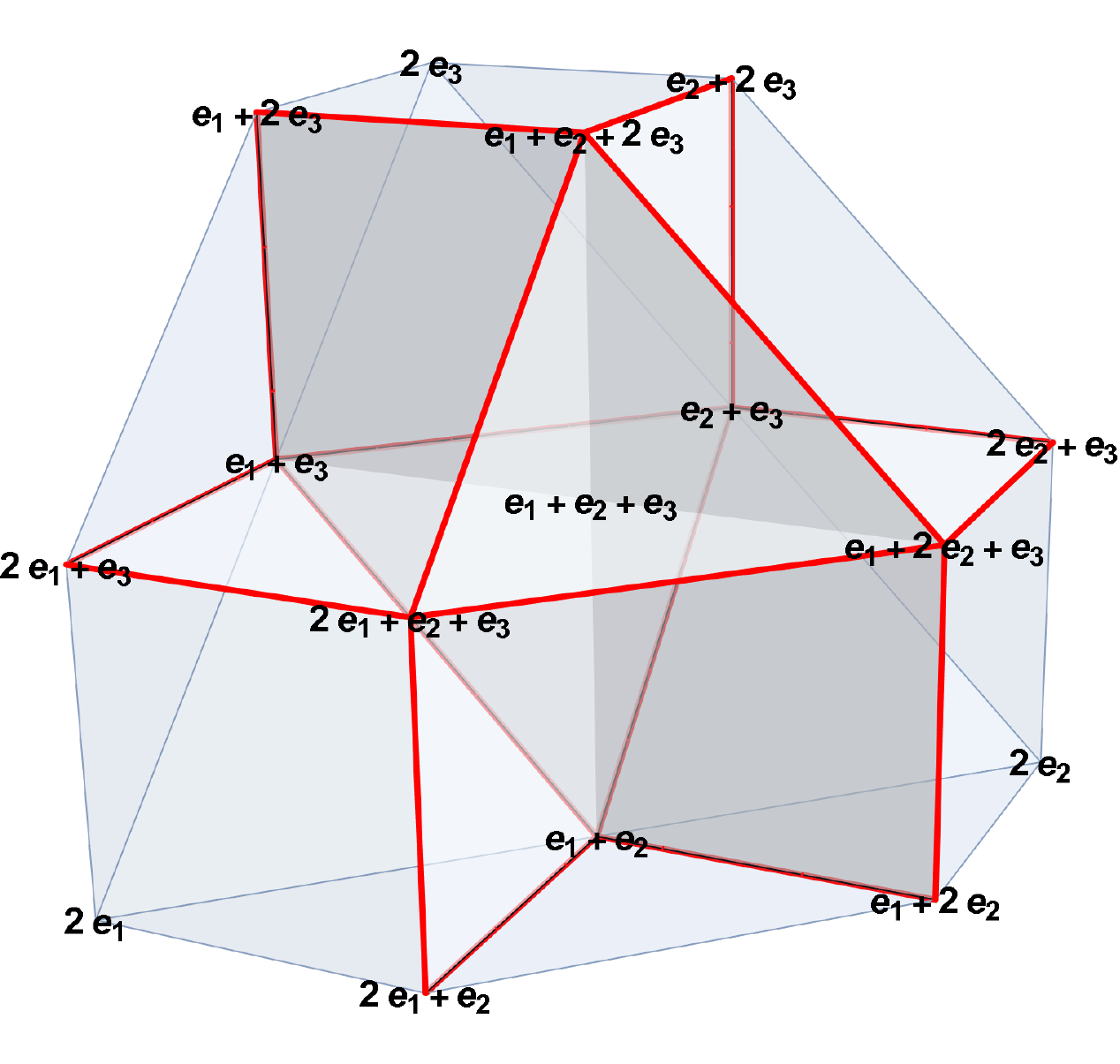}
		\caption{Image of any $\Delta_{2,n}$ under a generic projection $x\mapsto (x_i+\cdots +x_j,x_j+\cdots +x_k,x_k + \cdots +x_i)$ into $\mathbb{R}^3$.  Vertices of $\Delta_{2,n}$ map to the 16 lattice points, which in turn correspond exactly to the 16 connected components of Figure \ref{fig:connectedcomponentshypersimplexsplit158}.  The three lattice points $e_1+e_2,e_2+e_3,e_1+e_3$ correspond (via the projection) to the sets of Mandelstam invariants that are set to zero in split kinematics.  For emphasis, one of the three hexagonal cross-sections is highlighted.}
		\label{fig:splitkinematics3dprojection}
	\end{figure}
	One can easily see that the three hyperplanes 
	$$x_i+\cdots +x_j = 1,\ \ x_j+\cdots +x_k = 1,\ \ x_k+ \cdots+x_1 = 1,$$
	each of which subdivides $\Delta_{2,n}$ into two maximal cells, project onto the three hexagonal cross-sections in Figure \ref{fig:splitkinematics3dprojection}; moreover, as they are in generic position they subdivide $\Delta_{2,n}$ into eight maximal cells\footnote{These eight cells are not matroid polytopes, hence this is \textit{not} a matroid subdivision!  Note, however, that we are setting a large number of Mandelstams to zero, which ``blurs over'' the doubly subdivided octahedra.}, as can be seen in their projection into $\mathbb{R}^3$ via Figure \ref{fig:splitkinematics3dprojection}, into a cubical arrangement.  It would be natural to try to generalize the projection in Figure \ref{fig:splitkinematics3dprojection} to longer cyclic chains; the problem is left to future work.

	\subsection{Generalization to Other Theories}
	
	One of the most pressing questions is to find out if there are other theories with amplitudes that smoothly split. In this work only scalar theories that admit a CHY representation were considered. One of the key ingredients was the behavior of the matrix $\textbf{A}_n$ on the split kinematic subspace. There are other theories with CHY formulations based on the same matrix, such as the Born-Infeld theory. In such theories a new element is also present, it is a matrix that combines momenta and polarization vectors, known as $\Psi(p_a,\epsilon_a)$. It seems reasonable to expect that imposing conditions on the polarization vectors one could smoothly split such amplitudes. Of course, if the Pfaffian of $\Psi$ shows a good behavior then a whole new branch of theories could also smoothly split, such as Yang-Mills.  
	
	The attentive reader might have noticed that neither Born-Infeld nor Yang-Mills amplitudes can split solely in terms of currents within the corresponding theories as a degree (dimension) counting argument reveals. This means that currents outside the theories are needed. It is known that the Born-Infeld (BI) theory admits an extension in which BI photons interact with emergent YM gluons. It would be interesting to further explore this connection.

	\subsection{Relation to Causal Diamonds and the Soft-Limit Triangulation}
	
	A surprising connection between solutions to the wave equations and the space of planar Mandelstam invariants was uncovered in \cite{Arkani-Hamed:2019vag}. Properties of scattering amplitudes, such as factorization, can be translated into properties of the causal structure of an emergent space-time. 
	
	It is natural to consider what conditions on the causal structure are imposed on the $(i,j,k)$-split kinematic subpsace. Somehow the conditions that planar invariants which involve a chain of labels, which in the notation of \cite{Arkani-Hamed:2019vag} correspond to $X_{a,b} = s_{a,a+1,\ldots ,b,b+1}$ or $X_{a,b}=\eta_{a+1,b+1}$, can split into, e.g., $X_{a,b} = X_{a,i-1}+X_{i,b}$, must have a meaning in terms of how different regions interact with each other. It would be interesting to find a geometric interpretation of the semi-local property in this context. 
	
	In order to give more evidence that there are interesting connections, note that a recursion for biadjoint scattering amplitudes was presented in \cite{He:2018svj,Salvatori:2019phs,Arkani-Hamed:2019vag} using a novel soft-limit triangulation. For the reader's convenience we rewrite Equation 16 of \cite{Arkani-Hamed:2019vag} below,     
	\be\label{sofTriangle} 
	m_n = \sum_{i=4}^n \left(\frac{1}{X_{1,3}}+\frac{1}{X_{2,i}}\right){\hat m}_{n_L}\times {\hat m}_{n_R}. 
	\ee 
	In this equation the hatted amplitudes are the smaller amplitudes into which $m_n$ factors near the $X_{2,i}=0$ region with variables shifted so that $X_{2,j}\to X_{2,j}-X_{2,i}$. 
	
	Let us consider the $n=5$ and $n=6$ cases in order to show how degenerate 3-splits can naturally appear from \eqref{sofTriangle} by setting to zero all but one of the terms. The explicit form of \eqref{sofTriangle} for $n=5$ reads (see also \cite[Eq. 17]{Arkani-Hamed:2019vag}), 
	\be\label{wahu} 
	m_5(\mathbb{I},\mathbb{I}) = \left( \frac{1}{s_{12}}+\frac{1}{s_{23}} \right)\left( \frac{1}{s_{51}-s_{23}}+\frac{1}{s_{45}}\right) +  \left( \frac{1}{s_{12}}+\frac{1}{s_{51}} \right)\left( \frac{1}{s_{34}}+\frac{1}{s_{23}-s_{51}}\right) .
	\ee 
	Requiring the first term to vanish by setting the second factor to zero implies that we are exploring the subspace of kinematics space where $s_{23} = s_{45}+s_{51}$. Using that for $n=5$ $s_{23}=s_{451}$ we get the condition of a $(2,3,5)$-split which is a degenerate 3-split, i.e.
	$$s_{451} = s_{45}+s_{51}$$
	or $s_{41}=0$.
	
	Evaluating the second term in \eqref{wahu} on this subspace gives
	\be 
	\left. m_5(\mathbb{I},\mathbb{I})\right|_{\rm split\, kin.} = \left( \frac{1}{s_{12}}+\frac{1}{s_{51}} \right)\left( \frac{1}{s_{34}}+\frac{1}{s_{45}}\right) ={\cal J}(5,1,2){\cal J}(3,4,5).
	\ee 
	Of course, this is a degenerate 3-split because the third amputated current is trivial, i.e. ${\cal J}(2,3) =1$.
	
	Let us now consider the $n=6$ case. The formula \eqref{sofTriangle} becomes
	\be\label{pipo} 
	\begin{split}
		m_6(\mathbb{I},\mathbb{I}) = & \left( \frac{1}{s_{12}}+\frac{1}{s_{23}} \right){\hat m}(4,5,6,1,I) + \left( \frac{1}{s_{12}}+\frac{1}{s_{234}} \right){\hat m}(2,3,4,I){\hat m}(5,6,1,I)+\\ & \left( \frac{1}{s_{12}}+\frac{1}{s_{61}} \right){\hat m}(2,3,4,5,I).
	\end{split}
	\ee
	As explained in the definition of \eqref{sofTriangle} each hatted amplitude must be appropriately shifted and the meaning of the emergent particle $I$ is different in each term.
	
	Let us select kinematic invariants that set to zero the second and third terms in \eqref{pipo}. This is achieved by 
	\be 
	s_{234}= s_{23} + s_{34}, \quad  
	s_{2345} = s_{23} + s_{345}
	\ee
	which is clearly the $(6,1,3)$-split kinematic subspace, i.e., $s_{24}=s_{25}=0$. As expected, the first term in \eqref{pipo} gives the expected answer, i.e.
	\be\label{damo} 
	\left. m_6(\mathbb{I},\mathbb{I})\right|_{\rm split \, kin.} = \left( \frac{1}{s_{12}}+\frac{1}{s_{23}} \right){\cal J}(3,4,5,6).
	\ee

	A similar analysis shows that setting to zero the first and third terms in \eqref{pipo} by only imposing linear constrains leads to subspace in which the second term vanishes as well and therefore we do not get any interesting split.
	
	We have also considered each term in \eqref{pipo} evaluated on the $(1,3,5)$-split and $(2,4,6)$-split kinematic subspaces and found that the second term always vanishes while the other two are non-trivial functions which have to be added in order to exhibit the 3-split behavior.

	\subsection{CEGM Amplitudes: Connections and Prospects}\label{sec: CEGM}
	
	Let us now point out an intriguing similarity between the smooth splitting in Equation \eqref{example} and a particular residue of the generalized biadjoint scalar partial amplitudes $m^{(k)}_n(\mathbb{I},\mathbb{I})$, introduced by Cachazo, Early, Guevara and Mizera (CEGM) in \cite{Cachazo:2019ngv}. The CEGM construction starts with a generalization of the CHY formula for the biadjoint theory which is an integral over the space of $n$ marked points on $\mathbb{CP}^1$, also known as $X(2,n)$, to an integral over the space of $n$ marked points in $\mathbb{CP}^{k-1}$, i.e. $X(k,n)$. 
	
	The original motivation for the CEGM generalization came from the study of extensions of the combinatorial factorization procedure introduced in \cite{Cachazo:2017vkf} from sets of triples to sets of $(k+1)$-tuples, extensions of the delta algebra for MHV leading singularities to higher $k$ \cite{Cachazo:2018wvl}, and an effort to extend the cohomology ring of the moduli space $\text{Conf}_n(SU(2))\slash SU(2)$ \cite{early2019configuration} of $n$ points in $SU(2)$, to other moduli spaces, and in particular a certain combinatorial analog of the scattering equations which appears in the characterization of permutohedral blades \cite{Early:2018mac,2016arXiv161106640E}.

	The CEGM generalization of the CHY potential function is
	\be 
	{\cal S}^{(k)}_n := \sum_{j_1 < j_2 <\ldots < j_k} \mathfrak{s}_{j_1,j_2,\ldots  ,j_k} {\rm log}(|j_1,j_2,\ldots ,j_k|)
	\ee 
	where $|a_1,a_2,\ldots a_k|$ denote Plucker coordinates of $X(k,n)$. There are several important novelties in the theory, which we recall, for the reader's convenience.  First, the kinematic invariants for the theory are higher rank $k$ analogs of Mandelstam invariants $s_{i,j}$; they are indexed by $k$-element subsets, and we use the notation $\mathfrak{s}_{J} = \mathfrak{s}_{j_1,\ldots, j_k}$. Here, the generalization of masslessness is imposed by requiring $\mathfrak{s}_J$ be zero whenever an index is repeated. One also has the $n$ linear relations which generalize momentum conservation,
	$$\sum_{J \ni a} \mathfrak{s}_{J} = 0$$
	for each $a=1,\ldots, n$.
	
	In \cite{Cachazo:2019ngv}, the generalized biadjoint scalar $m^{(k)}_n(\mathbb{I},\mathbb{I})$ was constructed as follows
	\be\label{higherkBA} 
	m^{(k)}_n(\mathbb{I},\mathbb{I}) := \int \prod_{\alpha=1}^{(k-1)(n-k-1)}dx_\alpha\, \delta \left(\frac{\partial {\cal S}^{(k)}_n}{\partial x_\alpha}\right)\times ({\rm PT}^{(k)}(\mathbb{I}))^2
	\ee 
	where ${\rm PT}^{(k)}(\mathbb{I})$ is the $X(k,n)$ analog of the Parke-Taylor function  ${\rm PT}(\mathbb{I})$ presented in Eq. \eqref{PTdef} and $x_\alpha$ is some parameterization of $X(k,n)$. In the same way that the $k=2$ formula controls the leading order in an expansion around $\alpha'=0$ of string theory integrals, \eqref{higherkBA} has been shown to control the leading other in generalized string integrals \cite{Arkani-Hamed:2019mrd}.
	
	In order to present the connection with the smooth splitting of biadjoint amplitudes let us specialize to the case $k=3$ and $n=6$.
	
	Following \cite[Section 2.2]{Cachazo:2019ngv}, one finds that the kinematic invariant $\tilde{R}$, defined by 
	\begin{eqnarray}\label{eq:Rtilde}
		\tilde{R} = \mathfrak{s}_{156}+ \mathfrak{s}_{256}+ \mathfrak{s}_{345}+ \mathfrak{s}_{346}+ \mathfrak{s}_{356}+ \mathfrak{s}_{456},
	\end{eqnarray}
	is a pole of $m^{(3)}_6(\mathbb{I},\mathbb{I})$; it is the residue at $\tilde{R}=0$ that is now of interest.
	
	Now, in terms of the planar basis of kinematic invariants, introduced and developed by the second author in \cite{Early:2020hap,Early:2019eun,Early:2019zyi} in the context of permutohedral and hypersimplicial blades, Equation \eqref{eq:Rtilde} can be rewritten as $\tilde{R} = \eta_{246}(\mathfrak{s})$, where 
	\begin{eqnarray}\label{eq: planar basis expansion}
		\eta_{246} & = & \frac{1}{6}\left(6 \mathfrak{s}_{123}+5 \mathfrak{s}_{124}+4 \mathfrak{s}_{125}+3 \mathfrak{s}_{126}+4 \mathfrak{s}_{134}+3 \mathfrak{s}_{135}+2 \mathfrak{s}_{136}+2 \mathfrak{s}_{145}+\mathfrak{s}_{146} +6 \mathfrak{s}_{156}\right.\nonumber\\
		& + & \left. 3 \mathfrak{s}_{234}+2 \mathfrak{s}_{235}+\mathfrak{s}_{236}+\mathfrak{s}_{245}+5 \mathfrak{s}_{256}+6 \mathfrak{s}_{345}+5 \mathfrak{s}_{346}+4 \mathfrak{s}_{356}+3 \mathfrak{s}_{456} \right).
	\end{eqnarray}
	Here the coefficients have a precise meaning in combinatorial geometry in terms of certain regular matroid subdivisions of polytopes called \textit{hypersimplices}
	\begin{eqnarray}\label{eq:hypersimplex}
		\Delta_{k,n} = \left\{x\in \lbrack 0,1\rbrack^n: \sum_{j=1}^n x_j=k \right\}.
	\end{eqnarray}
	Specifically, each set of coefficients comes from the heights of the piecewise linear surface over $\Delta_{k,n}$, which projects down to induce the subdivision.  There are analogs of the formula in Equation \eqref{eq: planar basis expansion} for $m^{(k)}(\mathbb{I},\mathbb{I})$ for all $(k,n)$, with $2\le k\le n-2$ which are known \cite{Early:2020hap} to give rise to planar bases of kinematic invariants, introduced in  \cite{Early:2019eun} using a particular kind of tropical hypersurface called a blade, which enjoys a certain cyclic symmetry in $(1,2,\ldots, n)$.  In fact, when any given planar kinematic invariant $\eta_{j_1,\ldots, j_k}(\mathfrak{s})$ vanishes, one has a pole of $m^{(k)}_n$; this can be seen from a combinatorial perspective because the constant, rational coefficients of $\eta_{j_1,\ldots, j_k}$ induce a matroid subdivision of $\Delta_{k,n}$ that is \textit{coarsest}, that is, it is not the common refinement of any other collection of subdivisions.
	
	For the biadjoint scalar, which corresponds here to the case $k=2$, one recovers the planar kinematic invariants, as
	$$\eta_{ij} = s_{i+1\cdots j},$$
	and in fact one can interpret the usual identity $s_{J} = s_{J^c}$ using combinatorial principles, as the statement that two different surfaces over the hypersimplex $\Delta_{2,n}$ project down and induce the same matroid subdivision of it into the same pair of matroid polytopes.  Further, the massless condition that $p_i^2 = 0$ has the interpretation that the set of constant, rational coefficients of $\eta_{i,i+1}$ defines a height function over the vertices of $\Delta_{2,n}$ which does not bend over its interior.
	
	For example, the two blade arrangements in Figure \ref{fig:bladestetrahedronsquaremove} are dual to the Mandelstam invariants $s_{12}$ and $s_{23}$, respectively, by using the formula in Equation \eqref{eq:planar basis element} on the six vertices of the octahedron in the center of the twice dilated tetrahedra.
	\begin{figure}[h!]
		\centering
		\includegraphics[width=1\linewidth]{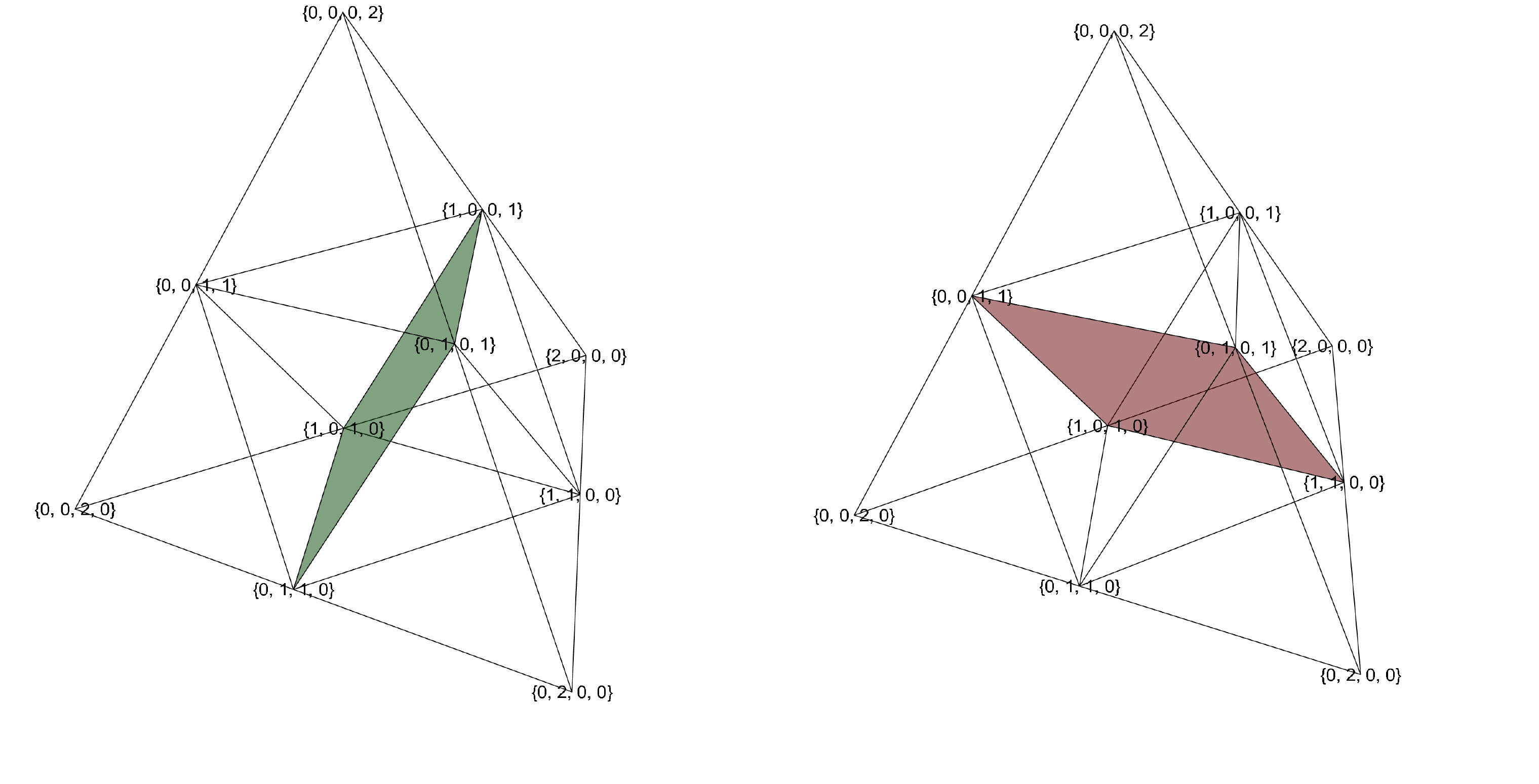}
		\caption{Combinatorial construction of the $s$ and $t$ channels using matroid subdivisions: lift each hyperplane to a surface in $\mathbb{R}^4$ which bends on exactly that hyperplane.  Left: $\eta_{24} = s_{12}$.  Right: $\eta_{13} = s_{23}$.}
		\label{fig:bladestetrahedronsquaremove}
	\end{figure}
	Such analysis leads to the correct generalization of the planar basis $s_{i\cdots j} = \sum_{i\le a<b\le j}s_{i,j}$, to CEGM generalized biadjoint amplitudes, where planar elements are labeled by cyclically non-contiguous (or, ``nonfrozen'') $k$-element subsets, of the form $\eta_{j_1,\ldots, j_k}(\mathfrak{s})$, where $\{j_1,\ldots, j_k\}$ is not a cyclic interval.  These are induced by certain coarsest matroid subdivisions of the hypersimplex $\Delta_{k,n}$.  
	
	The general formula for planar planar basis elements $\eta_{k_1,\ldots, j_k}(\mathfrak{s})$, which give rise to a particular subset of poles of $m^{(k)}_n$, in bijection with \textit{nonfrozen} $k$-element subsets  $\{j_1,\ldots, j_k\}$, is
	\begin{eqnarray}\label{eq:planar basis element}
		& & \eta_{j_1\cdots j_k}(\mathfrak{s}) \\
		& & = -\frac{1}{n}\sum_{I\in \binom{\lbrack n\rbrack}{k}}\min\{(e_{t+1}+2e_{t+2}+\cdots +(n-1)e_{t-1})\cdot (e_I-e_J): t=1,\ldots, n\}\mathfrak{s}_I,\nonumber
	\end{eqnarray}
	where we use the notation $e_J = \sum_{j\in J} e_j$.  That the $\eta_{j_1,\ldots, j_k}(\mathfrak{s})$ are linearly independent is nontrivial in general; this property was proved in \cite{Early:2020hap}.
	
	As a special case, from the definition one has, for $n=5$,
	$$\eta_{25} = \frac{1}{5} \left(4 s_{12}+3 s_{13}+2 s_{14}+s_{15}+2 s_{23}+s_{24}+5 s_{34}+4 s_{35}+3 s_{45}\right),$$
	which is equivalent modulo momentum conservation to the familiar expression
	$$\eta_{25} = s_{34} + s_{45} + s_{35}.$$
	
	Now we are ready to return to the phenomenon of splittings which has been observed to occur on residues for the generalized biadjoint scalar.  By \cite[Cor. 8]{Early:2019eun} the set of planar kinematic invariants $\eta_{abc}$ is a basis of linear functions on the kinematic space, so we can rewrite $m^{(3)}_6(\mathbb{I},\mathbb{I})$ as 
	\begin{eqnarray}\label{eq:36 amplitude Example GFDs}
		m^{(3)}_6(\mathbb{I},\mathbb{I}) & = & \frac{1}{\eta _{236} \eta _{246} \eta _{256} \eta _{346}} + \frac{\eta_{124} + \eta_{346} + \eta_{124}}{\eta_{124}\eta_{346}\eta_{246}(-\eta_{246} + \eta_{124} + \eta_{346} + \eta_{256})} \\
		& +&  \frac{1}{\eta_{256}\eta_{346}\eta_{356}(-\eta_{246} + \eta_{124} + \eta_{346} + \eta_{256})}+ \cdots,\nonumber
	\end{eqnarray}
	and one can check directly (see for instance \cite{Cachazo:2019ngv,Borges:2019csl,Early:2019eun}), that the residue of $m^{(3)}_6(\mathbb{I},\mathbb{I})$ at $\eta_{246}= 0$ is a product of three factors,
	\begin{eqnarray}\label{eq:3n3split}
		\text{Res}_{\eta_{246} = 0}(m^{(3)}_6(\mathbb{I},\mathbb{I})) & = &  \left(\frac{1}{\eta _{236}}+\frac{1}{\eta _{124}}\right) \left(\frac{1}{\eta _{256}}+\frac{1}{\eta _{146}}\right) \left(\frac{1}{\eta _{346}}+\frac{1}{\eta _{245}}\right).
	\end{eqnarray}
	
	Looking forward, we focus on an important outcome of this paper: we have established, using the CHY formalism, that the kind of novel behavior for residues of generalized CEGM amplitudes that has been observed in \cite{Cachazo:2019ngv}, with more progress in \cite{Arkani-Hamed:2019mrd,He:2020ray}, has an analog in three different Quantum Field Theories, as a semi-local ``shadow.'' 
	
	This shadow appears not only for the cubic scalar partial amplitude, but also for NLSM and, more surprisingly, the special Galileon amplitudes where a planar order is not present.  One of the most significant -- and intriguing -- contrasts is that the semi-local smooth 3-splits into amputated currents that we have explored in this paper do not occur at residues of the amplitude but on certain subspaces of the kinematic space where the amplitude does not have a singularity; but for $m^{(3)}_n(\mathbb{I},\mathbb{I})$ it has been observed directly to occur on residues where one (or more) compatible planar basis elements $\eta_{j_1j_2j_3}$ vanishes \cite{Early:2019eun}.  The 3-splitting behavior is not very well-understood, and in fact it remains a very pressing open question whether it continues to occur in any generality.  What lessons need to be learned here?
	
	Another interesting direction would be to study split kinematics in the context of likelihood geometry and in particular likelihood degenerations \cite{Sturmfels:2020mpv,Agostini:2021rze}.  It is natural to propose generalizations of split kinematics for higher rank $k\ge 3$; could one describe what happens to the solutions to the CEGM scattering equations as one approaches the split kinematics subspace, not only for $k=2$, but for $k=3$ and beyond?
	
	Finally, in \cite{Early:2021tce}, the second author proposed a \textit{linearly ordered} analog of the generalized biadjoint scalar $m^{(k)}_n(\mathbb{I},\mathbb{I})$, by in effect introducing a facet deformation of the PK polytope, introduced in \cite{Cachazo:2020wgu}, to a simple polytope, the \textit{PK associahedron}.  The conjecture formulated in \cite{Early:2021tce} amounts to the statement that the poset of compatible iterated residues should be anti-isomorphic to the noncrossing complex of k-element subsets of $\{1,\ldots, n\}$, as studied in \cite{santos2017noncrossing}.  From a combinatorial geometric perspective, two novelties developed in \cite{Early:2021tce} are an explicit realization of the PK associahedron as a Minkowski sum of Newton polytopes; the face poset here was conjectured to be anti-isomorphic to the noncrossing complex, which would establish a connection to \cite{santos2017noncrossing}, and an interpretation of its face poset in terms of the iterated residues of a rational function, that is, the generalized amplitude.  Moreover, also in \cite{Early:2021tce} a new binary geometry, as in \cite{Arkani-Hamed:2019plo}, with compatibility degree contained in a certain noncrossing complex $\mathbf{NC}_{k,n}$, was proposed for all $k\ge 3$.  Now one of the intriguing features of $\mathbf{NC}_{k,n}$ is that the crossing criterion necessarily has a \textit{linear} order $(1,2,\ldots, n)$ rather than a cyclic one, which is one of the main differences from $m^{(k)}(\mathbb{I},\mathbb{I})$.  Is there a physical interpretation of this restricted symmetry?  We do observe that such phenomena do appear in the context of the amputated currents which we consider here; however any possible connection is left to future investigation.
	
	Does the generalized amplitude in  \cite{Early:2021tce}, or the CEGM generalized biadjoint scalar, exhibit meaningful extensions of the smooth 3-split, either on some residue or smoothly?
	
	In particular, based on the results of this paper it is natural to expect that the 3-split residues observed for $m^{(3)}_6(\mathbb{I},\mathbb{I}),\ m^{(3)}_7(\mathbb{I}, \mathbb{I})$ will generalize; the development for the rank $k=3$ CEGM amplitude $m^{(3)}_n(\mathbb{I},\mathbb{I})$ and beyond is left to future work.  
	
	In the next section we sketch a promising direction for future research.
	
	\subsection{CEGM Amplitudes: Smooth Splits at $k=3$ }\label{sec: splitCEGM}
	
	Here we show how generalized $k=3$ amplitudes smoothly split when restricted to a kinematic subspace analogous to the one previously studied for Quantum Field Theory amplitudes.
	
	In order to study smooth splits in generalized amplitudes we will use the CEGM formulation \cite{Cachazo:2019ngv} introduced in Section \ref{sec: CEGM}. Without loss of generality, we consider the $k=3$ split kinematics subspace $(1,2,j,j+1)$ defined by setting to zero any $\mathfrak{s}_{abc}$ whose indices do not satisfy $1\leq a,b,c \leq j+1$ or $j\leq a,b,c \leq 2$, where the indices are understood modulo $n$.
	
	Due to the existing $SL(3,\mathbb{C})$ redundancy we can fix four particles, and a natural choice is the gauge fixing
	
	\begin{equation}
		\begin{blockarray}{cccc}
			\hspace{1mm}\textcolor{Maroon}{1} &  \textcolor{Maroon}{2} &  \textcolor{Maroon}{j} &  \textcolor{Maroon}{j+1} \\
			\begin{block}{[cccc]}
				\hspace{1mm}0 & 0 & 1 & 1 \\
				\hspace{1mm}0 & 1 & 0 & 1 \\
				\hspace{1mm}1 & 0 & 0 & 1 \\
			\end{block}
		\end{blockarray}
		\label{matrixk3}
	\end{equation}
	where punctures $2$ and $j$ are sent to infinity.
	
	Let us however start by writing the $k=3$ CEGM formula for punctures $1$, $2$, $j$ and $j+1$ fixed to generic values
	
	\begin{equation}\label{CEGMk3}
		m_n^{(3)}(\mathbb{I},\mathbb{I})=\int \prod_{a=1}^n\! ' \prod_{t=1}^2 dx_{t,a} \delta\left(\frac{\partial {\cal S}^{(3)}}{\partial x_{t,a}}\right) (V_{1,2,j,j+1}\,\textrm{PT}^{(3)}_n(\mathbb{I}))^2
	\end{equation}
	where $V_{1,2,j,j+1}\equiv |1,2,j||2,j,j+1||j,j+1,1||j+1,1,2|$ and the prime in the product means $a\not\in\{1,2,j,j+1\}$. The $k=3$ Parke-Taylor function is given by
	
	$$\textrm{PT}^{(3)}_n(\mathbb{I})=\frac{1}{|123||234|\cdots |n12|}\,.$$
	Using the gauge fixing \eqref{matrixk3} the factor $V_{1,2,j,j+1}\,\textrm{PT}^{(3)}_n(\mathbb{I})$ can be written as the product
	
	\begin{equation}
		\begin{split}
			&  \overbrace{\left(\frac{|12j||2,j,j+1||j,j+1,1||j+1,1,2|}{|123||234|\cdots |j-2,j-1,j||j-1,j,j+1||j,j+1,1||j+1,1,2|}\right)}^{V_{1,2,j,j+1}\,\textrm{PT}^{(3)}_{(1,2,...,j,j+1)}}\\ & \times\underbrace{\left(\frac{|j,j+1,1||j+1,1,2||12j||2,j,j+1|}{|j,j+1,j+2||j+1,j+2,j+3|\cdots |n-1,n,1||n12||12j||2,j,j+1|}\right)}_{V_{j,j+1,1,2}\,\textrm{PT}^{(3)}_{(j,j+1,...,n,1,2)}}\,,
		\end{split}
	\end{equation}
	where $V_{1,2,j,j+1}=V_{j,j+1,1,2}$. The first factor corresponds to the Parke-Taylor function for a generalized amplitude with the double ordering $(1,2,...,j+1)$ multiplied by the Fadeev-Popov factor $V_{1,2,j,j+1}$ that appears from the fixing of punctures $1$, $2$, $j$ and $j+1$. Similarly, the second factor corresponds to the Parke-Taylor function for the double ordering $(j,j+1,...,n,1,2)$ multiplied by the same Fadeev-Popov factor. Notice that the variables in each factor and after the gauge fixing \eqref{matrixk3} have completely decoupled.
	
	Now let us have a look at the $k=3$ CEGM potential 
	
	$${\cal S}_n^{(3)}:=\sum_{1\leq a<b<c\leq n}\mathfrak{s}_{abc}\,\textrm{log}|abc|$$
	in the split kinematics subspace $(1,2,j,j+1)$. Note that in this kinematics the potential splits into
	
	$${\cal S}_n^{(3)}={\cal S}^{(3)}_{j+1}+{\cal W}$$
	where the first term 
	
	$${\cal S}^{(3)}_{j+1}:=\sum_{1\leq a<b<c\leq j+1}\mathfrak{s}_{abc}\,\textrm{log}|abc|$$
	is the $k=3$ CEGM potential for a generalized amplitude with particles $(1,2,...,j+1)$ and the second term ${\cal W}$ is an object that we still have to identify. Let us look at it in more detail. This term can be written as
	
	\begin{equation}
		\begin{split}
			{\cal W}:=& \sum_{j\leq a<b<c\leq 2}\mathfrak{s}_{abc}\,\textrm{log}|abc|-\mathfrak{s}_{12j}\textrm{log}|12j|-\mathfrak{s}_{1,2,j+1}\textrm{log}|1,2,j+1|-\mathfrak{s}_{1,j,j+1}\textrm{log}|1,j,j+1|\\
			& -\mathfrak{s}_{2,j,j+1}\textrm{log}|2,j,j+1|
		\end{split}
	\end{equation}
	where the indices in the sum are understood modulo $n$. After using the gauge fixing \eqref{matrixk3} we have
	
	$$\textrm{log}|12j|=\textrm{log}|1,2,j+1|=\textrm{log}|1,j,j+1|=\textrm{log}|2,j,j+1|=0$$
	and the variables in the two terms ${\cal S}_{j+1}^{(3)}$ and ${\cal W}$ completely decouple. Moreover, we can now identify ${\cal W}|_{\eqref{matrixk3}}$ with the $k=3$ CEGM potential for a generalized amplitude with particles $(j,j+1,...,n,1,2)$, i.e.
	
	$${\cal W}|_{\eqref{matrixk3}}\equiv{\cal S}^{(3)}_{(j...n12)}|_{\eqref{matrixk3}}\,.$$
	Putting all the pieces together one can see that with the gauge fixing \eqref{matrixk3} the CEGM integral \eqref{CEGMk3} under the split kinematics $(1,2,j,j+1)$ splits into 
	\begin{equation}
		\begin{split}
			m_n^{(3)}(\mathbb{I},\mathbb{I})|_{(1,2,j,j+1)}& = \left(\int \prod_{a=3}^{j-1} \prod_{t=1}^2 dx_{t,a} \delta\left(\frac{\partial {\cal S}_{j+1}^{(3)}|_{\eqref{matrixk3}}}{\partial x_{t,a}}\right) (V_{1,2,j,j+1}\,\textrm{PT}^{(3)}_{(1,2,...,j,j+1)})^2|_{\eqref{matrixk3}}\right)\\
			& \times\left(\int \prod_{a=j+2}^{n} \prod_{t=1}^2 dx_{t,a} \delta\left(\frac{\partial {\cal S}_{(j...n12)}^{(3)}|_{\eqref{matrixk3}}}{\partial x_{t,a}}\right) (V_{j,j+1,1,2}\,\textrm{PT}^{(3)}_{(j,...,n,1,2)})^2|_{\eqref{matrixk3}}\right)
		\end{split}
	\end{equation}
	where from \eqref{CEGMk3} one can see that the first factor is identified with an object that resembles the generalized amplitude $m^{(3)}_{j+1}(\alpha_1,\alpha_1)$ with $\alpha_1=(1,2,...,j+1)$, while the second factor is identified with an object that resembles the generalized amplitude $m^{(3)}_{n-j+3}(\alpha_2,\alpha_2)$ with $\alpha_2=(j,...,n,1,2)$. However, these two factors in the split are not amplitudes since their particles do not satisfy momentum conservation. We leave the interpretation of these resulting objects for future research.
	
	\section*{Acknowledgements}
	
	The second author is very grateful to the Institute for Advanced Study for excellent working conditions while this project was initiated, and he thanks Nima Arkani-Hamed, Johannes Henn and Bernd Sturmfels for related discussions and encouragement.  This research was supported in part by a grant from the Gluskin Sheff/Onex Freeman Dyson Chair in Theoretical Physics and by Perimeter Institute. Research at Perimeter Institute is supported in part by the Government of Canada through the Department of Innovation, Science and Economic Development Canada and by the Province of Ontario through the Ministry of Colleges and Universities.  This research received funding from the European Research Council (ERC) under the European Union’s Horizon 2020 research and innovation programme (grant agreement No 725110), Novel structures in scattering amplitudes.
	
	\appendix

	\section{Definition of Amputated Currents}\label{sec: currents}
	
	Throughout this work, we have used amputated currents in various quantum field theories of scalars in order to characterize the behaviour of the corresponding amplitudes when restricted to the split kinematic subspace. In this appendix we give a formal definition of the objects. 
	
	Currents are objects in quantum field theory which appear when one interpolates between correlation functions and scattering amplitudes. Recall that the LSZ formalism starts with a correlation function of operators in coordinate space $G(x_1,x_2,\ldots ,x_n)$. Fourier transforming to momentum space produces a distribution localized on the momentum conservation loci
	$$ \delta^D(p_1+p_2+\cdots +p_n)\tilde G(p_1,p_2,\ldots ,p_n)\,. $$
	This is due to translational invariance of the correlation function $G(x_1,x_2,\ldots ,x_n)$.
	
	The function $\tilde G(p_1,p_2,\ldots ,p_n)$ has simple poles of the form $1/p_i^2$ and a scattering amplitude is obtained by the limiting procedure (or multidimensional residue computation)
	\be\label{LSZ}
	A(p_1,p_2,\ldots ,p_n) = \left(\prod_{i=1}^n\lim_{p_i^2\to 0} p_i^2\right) \tilde G(p_1,p_2,\ldots ,p_n).
	\ee
	The process of multiplying by $p_i^2$ is called ``amputating" the $i^{\rm th}$-leg. 
	A current is defined by performing all but one of the operations in \eqref{LSZ}. Let us assume that the $n^{\rm th}$-leg is spared. Then,
	\be\label{currentA}
	J(p_1,p_2,\ldots ,p_{n-1}) := \left(\prod_{i=1}^{n-1}\lim_{p_i^2\to 0} p_i^2\right) \tilde G(p_1,p_2,\ldots ,p_{n-1},p_n).
	\ee
	Note that the current still possesses the $1/p_n^2$ pole and hence the $n^{\rm th}$ leg is said to remain off-shell, i.e. $p^2_n\neq 0$.
	
	In this work, the relevant object is the amputated current, i.e.
	\be\label{currentB}
	{\cal J}(p_1,p_2,\ldots ,p_{n-1}) := p_n^2 J(p_1,p_2,\ldots ,p_{n-1}) .
	\ee

	In general, (amputated) currents are not unique. This is most apparent in gauge theories where currents are not even gauge invariant. The reason is that physical observables are obtained from scattering amplitudes and therefore any two currents that differ by something that vanishes when $p_n^2=0$ lead to the same physical consequences. 
	
	Here, however, we are using currents to determine the behavior of amplitudes and as such there can be no ambiguity. 
	
	Luckily, for scalar theories there is a natural prescription which provides the required definition. The Feynman diagrams used to compute correlation functions in momentum space and amplitudes are combinatorially identical. The prescription is to write each Feynman diagram in terms of a basis of Mandelstam invariants provided by the planar ones with respect to the canonical order $\mathbb{I}$. Each such invariant can be made to depend on only a set of particles not containing label $n$. Each Feynman diagram is then fully amputated.  
	
	While this definition is precise, it is not very effective in practice as computing amplitudes or currents using Feynman diagrams quickly becomes impractical as $n$ increases. This is why we provide a definition using the CHY formalism. In fact, this definition leads exactly to the amputated currents that appear in smooth splittings. 
	
	Consider the most general CHY potential for $n$ particles and we will allow three of them to be off-shell, say particles $i,j,k$. Of course, we are only interested in the case with a single off-shell particles but the construction is more uniform is we allow all three to be off-shell. 
	
	Following Naculich's construction \cite{Naculich:2015zha}, we define the modified CHY potential\footnote{Naculich works directly with the scattering equations and not with the potential but it is straightforward to translate.}
	\be
	\begin{split}\label{naculich}
		{\cal S}_n = & \sum_{a<b} 2p_a\cdot p_b\log \, (\sigma_a-\sigma_b) + (p_i^2+p_j^2-p_k^2)\log\, (\sigma_i-\sigma_j)+ \\  & (p_k^2+p_i^2-p_j^2)\log\, (\sigma_k-\sigma_i)+(p_j^2+p_k^2-p_i^2)\log\, (\sigma_j-\sigma_k).
	\end{split}
	\ee
	
	Note that this potential was designed as to preserve $SL(2,\mathbb{C})$ invariance. This means that three of the punctures can be fixed and it is natural to take the set $\{ \sigma_i,\sigma_j,\sigma_k\} $ to be $\{ 0,1,\infty \}$. 
	
	Let us choose $\sigma_i =0,\sigma_j=1$, and $\sigma_k=\infty$. In this case the potential becomes
	\be\label{currentCHYP}
	\begin{split}
		{\cal S}_n = \sum_{a<b\, :\, a,b\notin \{ i,j,k \}} \!\!\! s_{ab}\log \, (\sigma_a-\sigma_b) +\sum_{a\notin \{i,j\}}\left( 2p_a\cdot p_i \log\, (\sigma_a) + 2p_a\cdot p_j \log\, (1-\sigma_a)\,\right).   
	\end{split}
	\ee
	Note that any term containing $\sigma_k$ drops out while $\log (\sigma_i-\sigma_j) = \log 1 =0$. 
	
	Having constructed the CHY potential it is possible to give the CHY formula for the five kind of amputated currents used in the main text. 
	
	We present them in the form of a lemma. In the lemma the CHY potential ${\cal S}_n$ is always the one defined in \eqref{currentCHYP}. We also use $\textbf{A}_n^{[p\, q]}$ to denote the submatrix of the matrix $\textbf{A}_n$ obtained by removing the $p^{\rm th}$ and $q^{\rm th}$ rows and columns. The entries of the $n\times n$ matrix $\textbf{A}_n$ that do not involve off-shell legs are given by the standard expression $A_{ab}=s_{ab}/(\sigma_a-\sigma_b)$. Likewise, $\textbf{A}_n^{[i\, j\, k]}$ denotes the submatrix of the matrix $\textbf{A}_n$ obtained by removing the $i^{\rm th}$, $j^{\rm th}$ and $k^{\rm th}$ rows and columns.
	
	Before proceeding, a comment on notation is required. An amputated current is often written in a form in which the $n^{\rm th}$ particle corresponds to the off-shell leg and to indicate this the $n^{\rm th}$ label is not shown as in \eqref{currentB}. However, in the statement of the lemma we allow the off-shell leg to be any leg in a given set and therefore all labels are shown in the currents.   
	
	\begin{lem}
		Let $q\in \{i,j,k\}$ represent the off-shell leg of the current. Then the CHY representation of a biadjoint amputated current is given by,
		\be 
		{\cal J}(1,2,\ldots ,n) = \int\!\!\prod_{a\notin \{i,j,k\}}\!\! d\sigma_a\delta\left(\frac{\partial {\cal S}_n}{\partial \sigma_a}\right) \left(\frac{|i\,j||j\, k||k\, i|}{|1\, 2||2\, 3|\cdots |n-1\,n||n\, 1|}\right)^2\,.
		\ee 
		The CHY representation of a NLSM amputated current is,
		\be 
		{\cal J}^{\rm NLSM}(1,2,\ldots ,n) = \int\!\!\prod_{a\notin \{i,j,k\}}\!\! d\sigma_a\delta\left(\frac{\partial {\cal S}_n}{\partial \sigma_a}\right) \left(\frac{|i\,j||j\, k||k\, i|}{|1\, 2||2\, 3|\cdots |n-1\,n||n\, 1|}\right) \, \frac{|i\,j||j\, k||k\, i|}{|p\, q|^2}{\rm det} \textbf{A}_n^{[p\, q]}.
		\ee 
		Here $q$ is arbitrary (with $q\neq p$), although in practise it is convenient to choose it in the set $\{i,j,k\}$.
		
		Similarly, the CHY representation of a mixed NLSM amputated current is given by,
		\be 
		{\cal J}^{{\rm NLSM}\oplus \phi^3}(1,2,\ldots ,n|i,j,k) = \int\!\!\prod_{a\notin \{i,j,k\}}\!\! d\sigma_a\delta\left(\frac{\partial {\cal S}_n}{\partial \sigma_a}\right) \left(\frac{|i\,j||j\, k||k\, i|}{|1\, 2||2\, 3|\cdots |n-1\,n||n\, 1|}\right) \, {\rm det} \textbf{A}_n^{[i\, j\, k]}.
		\ee 

		The CHY representation of a special Galileon amputated current is,
		\be 
		{\cal J}^{\rm sGal} = \int\!\!\prod_{a\notin \{i,j,k\}}\!\! d\sigma_a\delta\left(\frac{\partial {\cal S}_n}{\partial \sigma_a}\right) \left( \frac{|i\,j||j\, k||k\, i|}{|p\, q|^2}{\rm det} \textbf{A}_n^{[p\, q]}\right)^2
		\ee 
		and finally the CHY representation of a mixed special Galileon amputated current is,
		\be 
		{\cal J}^{{\rm sGal}\oplus \phi^3}(i,j,k) = \int\!\!\prod_{a\notin \{i,j,k\}}\!\! d\sigma_a\delta\left(\frac{\partial {\cal S}_n}{\partial \sigma_a}\right) \left( {\rm det} \textbf{A}_n^{[i\, j\, k]}\right)^2.
		\ee 
	\end{lem}
	
	\begin{proof}
		
		To prove the lemma it is required to show that the corresponding CHY formulas reproduce the amputated currents as defined using Feynman diagrams. However, for scalar field theories, this is evident from the Dolan-Goddard proof of biadjoint amplitudes \cite{Dolan:2013isa} and from Naculich's general construction \cite{Naculich:2015zha}.
		
	\end{proof}

	\section{Proof of Determinantal Product Formula: Lemma \ref{detLemma}} \label{apA}
	
	In the main text we proved the smooth splitting formula for NSLM and special Galileon amplitudes using Lemma \ref{detLemma}. In this appendix we provide the proof. For the reader's convenience we rewrite the statement of the Lemma.
	
	\begin{lem}
		Let $M\in \mathbb{C}^{2m\times 2m}$ be antisymmetric, $L\in \mathbb{C}^{r\times r}$, and $W\in \mathbb{C}^{(2m+r)\times (2m+r)}$ defined in terms of $M$ and $L$ as follows
		\begin{equation}
			W :=
			\left[ \begin{array}{ccccc|ccccc}
				0 & M_{1,2} & \cdots & M_{1,2m-1}  & M_{1,2m} &  0 & 0 & 0 & \cdots & 0  \\
				-M_{1,2} & 0 & \cdots & M_{2,2m-1}  & M_{2,2m} & 0 & 0 & 0 & \cdots & 0  \\
				\vdots & \vdots & \vdots & \ddots & \vdots & \vdots & \vdots & \vdots & \ddots & \vdots \\
				-M_{1,2m-1} & -M_{2,2m-1} & \cdots & 0  & M_{2m-1,2m} & 0 & 0 & 0 & \cdots & 0  \\
				-M_{1,2m} & -M_{2,2m} & \cdots & -M_{2m-1,2m}  & 0 & c_1 & c_2 & c_3 & \cdots & c_{r} \\
				\hline
				0 & 0 & \cdots & 0 & d_1 & L_{1,1} & L_{1,2} & \cdots & L_{1,r-1} & L_{1,r}  \\
				0 & 0 & \cdots & 0 & d_2 & L_{2,1} & L_{2,2} & \cdots & L_{2,r-1} & L_{2,r} \\
				\vdots & \vdots & \ddots & \vdots & \vdots & \vdots & \vdots & \ddots & \vdots & \vdots \\
				0 & 0 & \cdots & 0 & d_r & L_{r,1} & L_{r,2} & \cdots & L_{r,r-1} & L_{r,r}\\
			\end{array}
			\right] 
		\end{equation}
		with $d_a$ and $c_a$ arbitrary complex numbers, then the following holds
		\be  
		\textrm{det}(W) = \textrm{det}(M)\textrm{det}(L).
		\ee\label{stat}
		
	\end{lem}

	\begin{proof}
		Let us compute the determinant on the LHS of (\ref{stat}) using the $2m$-th column to expand. Note that the contribution from any $d_a$ is of the form
		\begin{equation}
			\textrm{det}\left[\begin{array}{c|c}
				P & Q \\
				\hline
				0 & R
			\end{array}\right]=\textrm{det}(P)\textrm{det}(R)
			\label{detPQR}
		\end{equation}
		where 
		\begin{equation}
			P=\left[\begin{array}{ccccc}
				0 & M_{1,2} & M_{1,3} & \cdots & M_{1,2m-1} \\
				-M_{1,2} & 0 & M_{2,3} & \cdots & M_{2,2m-1} \\
				-M_{1,3} & -M_{2,3} & 0 & \cdots & M_{3,2m-1} \\
				\vdots & \vdots & \vdots & \ddots & \vdots  \\
				-M_{1,2m-1} & -M_{2,2m-1} & -M_{3,2m-1} & \cdots & 0 \\
			\end{array}\right]\,.
		\end{equation}
		Since $P$ is an odd-dimensional antisymmetric matrix, its determinant is zero and therefore the determinant (\ref{detPQR}) vanishes. This implies that the determinant on the LHS of (\ref{stat}) is independent of $d_a$. Likewise, the determinant can also be shown to be independent of $c_a$. 
		
		Having proved that \eqref{detPQR} is independent of the values of $d_a$ and $c_a$, it is possible to set them to any convenient values. In this case, it is clear that by setting $d_a=c_a =0$ for all $a\in \{ 1,2,\ldots ,r\}$ one is left with the determinant of a block diagonal matrix. Using the elementary property of determinants that the determinant of a block-diagonal matrix is the product of the determinants of the blocks the result follows. 
	\end{proof}
	
	\bibliographystyle{jhep}
	\bibliography{references}
	
\end{document}